\documentclass[11pt]{llncs}
\usepackage{epsfig,latexsym}
\usepackage{xcolor}

\newcommand{\WW}{$W_6^{++}$}

\usepackage{color}
\newcommand{\commentout}[1]{}

\begin{document}
\title{$\alpha_i$-Metric Graphs: Radius, Diameter and all Eccentricities
\thanks{This work was supported by a grant of the Romanian Ministry of Research, Innovation and Digitalization, CCCDI - UEFISCDI, project number PN-III-P2-2.1-PED-2021-2142, within PNCDI III.}}

\author{Feodor F. Dragan\inst{1} \and
Guillaume Ducoffe\inst{2}  }

\institute{Computer Science Department, Kent State University, Kent, USA 
\email{dragan@cs.kent.edu}  
\and 
National Institute for Research and
Development in Informatics and \\ University
of Bucharest, Bucure\c{s}ti, Rom\^{a}nia 
\email{guillaume.ducoffe@ici.ro} 
}

\maketitle

\begin{abstract} 
We extend known results on chordal graphs and distance-hereditary graphs to much larger graph classes by using only a common metric property of these graphs. Specifically, a graph is called $\alpha_i$-metric ($i\in \mathcal{N}$) if it satisfies the following $\alpha_i$-metric property for every vertices $u,w,v$ and $x$: if a shortest path between $u$ and $w$ and a shortest path between $x$ and $v$ share a terminal edge $vw$, then $d(u,x)\geq d(u,v) + d(v,x)-i$. Roughly, gluing together any two shortest paths along a common terminal edge may not necessarily result in a shortest path but yields a ``near-shortest'' path with defect at most $i$. It is known that $\alpha_0$-metric graphs are exactly ptolemaic graphs, and that chordal graphs and distance-hereditary graphs are $\alpha_i$-metric for $i=1$ and $i=2$, respectively. We show that an additive $O(i)$-approximation of the radius, of the diameter, and in fact of all vertex eccentricities of an $\alpha_i$-metric graph can be computed in total linear time. Our strongest results are obtained for $\alpha_1$-metric graphs, for which we prove that a central vertex can be computed in subquadratic time, and even better in linear time for so-called $(\alpha_1,\Delta)$-metric graphs (a superclass of chordal graphs and of plane triangulations with inner vertices of degree at least $7$). The latter answers a question raised in (Dragan, {\it IPL}, 2020). Our algorithms follow from new results on centers and metric intervals of $\alpha_i$-metric graphs. In particular, we prove that the diameter of the center is at most $3i+2$ (at most $3$, if $i=1$). The latter partly answers a question raised in (Yushmanov \& Chepoi, {\it Mathematical Problems in Cybernetics}, 1991).
\medskip

{\it Keywords:} metric graph classes; chordal graphs; $\alpha_i$-metric; radius; diameter; vertex eccentricity; eccentricity approximating trees; approximation algorithms.
\end{abstract}


\section{Introduction} 

Euclidean spaces have the following nice property: if the geodesic between $u$ and $w$ contains $v$, and the geodesic between $v$ and $x$ contains $w$, then their union must be the geodesic between $u$ and $x$. 
In 1991, Chepoi and Yushmanov introduced $\alpha_i$-metric properties ($i \in \mathcal{N}$), as a way to quantify by how much a graph is close to satisfy this above requirement~\cite{YuCh1991} (see also \cite{Ch86,Ch88} for earlier use of $\alpha_1$-metric property).
All graphs $G=(V,E)$ occurring in this paper are connected, finite, unweighted, undirected, loopless and without multiple edges.  
The \emph{length} of a path between two vertices $u$ and~$v$ is the number of edges in the path.  
The \emph{distance}~$d_G(u,v)$ is the length of a shortest path connecting $u$ and~$v$ in $G$.
The {\sl interval} $I_G(u,v)$ between $u$ and $v$ consists of all vertices on shortest $(u,v)$-paths, that is, it consists of all vertices (metrically) between $u$ and $v$: $I_G(u,v)=\{x \in V : d_G(u,x) + d_G(x,v) = d_G(u,v)\}.$
Let also $I_G^o(u,v)=I_G(u,v)\setminus \{u,v\}$. 
If no confusion arises, we will omit subindex {\small $G$}. 
\begin{description} 
\item[$\alpha_i$-metric property:] if $v \in I(u,w)$ and $w\in I(v,x)$ are adjacent,
  then \\  $~~~~~~~~~~~~~~~~~~~~~~~~~~~d(u,x)\geq d(u,v) + d(v,x)-i=d(u,v) + 1 + d(w,x)-i.$
\end{description} 
Roughly, gluing together any two shortest paths along a common terminal edge may not necessarily result in a shortest path (unlike in the Euclidean space) but yields a ``near-shortest'' path with defect at most $i$. 
A graph is called $\alpha_i$-metric if it satisfies the $\alpha_i$-metric property.
$\alpha_i$-Metric graphs were investigated in \cite{Ch86,Ch88,YuCh1991}.  
In particular, it is known that $\alpha_0$-metric graphs are exactly the distance-hereditary chordal graphs, also known as ptolemaic graphs~\cite{Hov77}. 
Furthermore, $\alpha_1$-metric graphs contain all chordal graphs~\cite{Ch86} and all plane triangulations with inner vertices of degree at least $7$~\cite{WG16}. 
$\alpha_2$-Metric graphs contain all distance-hereditary graphs~\cite{YuCh1991} and, even more strongly, all HHD-free graphs~\cite{ChD03}.
Evidently, every graph is an $\alpha_i$-metric graph for some $i$.  
Chepoi and Yushmanov in \cite{YuCh1991} also provided a characterization of all $\alpha_1$-metric graphs: They are exactly the graphs where all disks are convex and the graph \WW\ from Fig. \ref{fig:forbid} is forbidden as an isometric subgraph (see \cite{YuCh1991} or Theorem \ref{th:charact}). This nice characterization was heavily used in \cite{BaCh2003} in order to characterize $\delta$-hyperbolic graphs with $\delta\le 1/2$. 

The \emph{eccentricity}~$e_G(v)$ of a vertex~$v$ in $G$ is defined by $\max_{u \in V} d_G(u, v)$, i.e., it is the distance to a most distant vertex. 
The {\em diameter} of a graph is the maximum over the eccentricities of all vertices: $diam(G) = \max_{u\in V} e_G(u)=\max_{u,v\in V} d_G(u,v)$.
The {\em radius} of a graph is the minimum over the eccentricities of all vertices: $rad(G) = \min_{u\in V}e_G(u)$.
The {\em center} $C(G)$ of a graph $G$ is the set of its vertices with minimum eccentricity, i.e., $C(G) = \{u\in V : e_G(u) =rad(G)\}$. Each vertex from $C(G)$ is called a {\em central} vertex. 
In this paper, we investigate the radius, diameter, and all eccentricities computation problems in $\alpha_i$-metric graphs.
Understanding the eccentricity function of a graph and being able to efficiently compute or estimate the diameter, the radius, and all vertex eccentricities is of great importance. 
For example, in the analysis of social networks (e.g., citation networks or recommendation networks), biological systems (e.g., protein interaction networks), computer networks (e.g., the Internet or peer-to-peer networks), transportation networks (e.g., public transportation or road networks), etc., the eccentricity of a vertex is used in order to measure its importance in the network:
the {\em eccentricity centrality index} of $v$ \cite{Brandes} is defined as $\frac{1}{e(v)}$. 
Furthermore, the problem of finding a central vertex is one of the most famous facility location problems in Operation Research and in Location Science. 
In~\cite{YuCh1991},  the following nice relation between the diameter and the radius of an $\alpha_i$-metric graph $G$ was established: $diam(G)\ge 2rad(G)-i-1$. Recall that for every graph $G$, $diam(G)\leq 2rad(G)$ holds. Authors of \cite{YuCh1991} also raised a question\footnote{It is conjectured in~\cite{YuCh1991} that $diam(C(G))\leq i+2$ for every  $\alpha_i$-metric graph $G$.} whether the diameter of the center of an $\alpha_i$-metric graph can be bounded by a linear function of $i$.  It is known that the diameters of the centers of chordal graphs  or of distance-hereditary graphs are at most 3~\cite{Ch88,YuCh1991}.  

\medskip
\noindent
{\bf Related work on computing or estimating the radius, diameter, or all eccentricities.}
A naive algorithm which runs a BFS from each vertex to compute its eccentricity and then (in order to compute the radius, the diameter and a central vertex) picks the  smallest and the largest eccentricities and a vertex with smallest eccentricity has running time ${O}(nm)$ on an $n$-vertex $m$-edge graph. Interestingly, this naive algorithm is conditionally optimal for general graphs as well as for some restricted families of graphs \cite{AVW16,BoCrHa16,Chepoi2018FastEA,RoV13} since, under plausible complexity assumptions, neither the diameter nor the radius can be computed in truly subquadratic time (i.e., in ${O}(n^am^b)$ time, for some positive $a, b$ such that $a+b <2$) on those graphs. Already for split graphs (a subclass of chordal graphs), computing the diameter is roughly equivalent to {\sc Disjoint Sets}, {\it a.k.a.}, the monochromatic {\sc Orthogonal Vector}   problem~\cite{ChD92}. Under the Strong Exponential-Time Hypothesis (SETH), we cannot solve {\sc Disjoint Sets} in truly subquadratic time, and so neither we can compute the diameter of split graphs in truly subquadratic time~\cite{BoCrHa16}.

In a quest to break this quadratic barrier (in the size $n+m$ of the input), there has been a long line of work presenting more efficient algorithms for computing the diameter and/or the radius, or even better all eccentricities, on some special graph classes, by exploiting their geometric and tree-like representations and/or some forbidden pattern ({\it e.g.}, excluding a minor~\cite{DHV19+a}, or a family of induced subgraphs). 
For example, faster algorithms for all eccentricities computation are known for distance-hereditary graphs~\cite{CDP18,Dragan94-swat,DraganG20-tcs,DraganN00-dam}, outerplanar graphs~\cite{FaP80}, planar graphs~\cite{Cab18,GKHM+18}, graphs with bounded tree-width~\cite{AVW16,BHM18,DHV19+a} and, more generally, graphs with bounded clique-width~\cite{CDP18,Duc2020-CoRR}.
Linear-time algorithms for computing all eccentricities are also known for interval graphs~\cite{DNB97,Ola90}.
Some recent works have further studied which properties of interval graphs could imply on their own faster algorithms for diameter and all eccentricities computation.
Efficient algorithms for these problems have been found for AT-free graphs~\cite{Duc20-CoRR}, directed path graphs~\cite{CDHP01}, strongly chordal  graphs~\cite{DraganPhD}, dually chordal graphs~\cite{BrandstadtCD98-dam,Dragan1993-CSJofM}, Helly graphs and graphs of bounded Helly number~\cite{DrDuGu21-wads,Duc23,DuDr21-netw}.
See also~\cite{Duc21-CoRR}. 
Chordal graphs are another well-known generalization of interval graphs.
Although the diameter of a split graph can unlikely be computed in subquadratic time, there is an elegant linear-time algorithm for computing the radius and a central vertex of a chordal graph \cite{ChDr94}. 
However, until this work there has been little insight about how to extend this nice result to larger graph classes (a notable exception being the work in~\cite{ChD03}).
This intriguing question is partly addressed in our paper.

Since the exact diameter or radius computation in subquadratic time is impossible (unless the SETH is false) even for simple families of graphs (in case of the diameter, even for split graphs), a large volume of work was also devoted to approximation algorithms. 
It is known \cite{Che++2014} that the diameter of any graph with $n$ vertices and $m$ edges can be approximated within a multiplicative factor of $3/2$ in $\tilde{O}(m^{3/2})$ time. Furthermore, unless the SETH is false, no $O(n^{2-\epsilon})$ time algorithm can achieve an approximation factor better than $3/2$ in sparse graphs~\cite{RoV13} and no $O(m^{3/2-\epsilon})$ time algorithm can achieve an approximation factor better than $5/3$ \cite{BRSVW-stoc18}. 
The eccentricities of all vertices of any graph can be approximated within a factor of $5/3$ in $\tilde{O}(m^{3/2})$ time~\cite{Che++2014} and, under the SETH, no $O(n^{2-\epsilon})$ time algorithm can achieve better than $5/3$ approximation in sparse graphs \cite{AVW16} and no $O(m^{3/2-\epsilon})$ time algorithm can  achieve an approximation factor better than $9/5$~\cite{BRSVW-stoc18}. Authors of~\cite{BRSVW-stoc18} also show that no near-linear time algorithm can achieve a better than 2 approximation for the eccentricities and provide an algorithm that approximates eccentricities within a $2+\epsilon$ factor in $\tilde{O}(m/\epsilon)$ time for any $0<\epsilon<1$.  On planar graphs, there is an approximation scheme with near linear running time \cite{WeYu16}.   
Authors of \cite{Che++2014} additionally address a more challenging question of obtaining an additive $c$-approximation for the diameter $diam(G)$ of a graph $G$, i.e., an estimate $D$ such that $diam(G)- c\le D \le diam(G)$. 
A simple 
$\tilde{O}(m n^{1-\epsilon})$ 
time algorithm achieves an additive $n^{\epsilon}$-approximation and,  for any $\epsilon > 0$, getting an additive $n^{\epsilon}$-approximation algorithm for the diameter running in $O(n^{2-\epsilon'})$ time for any $\epsilon' > 2\epsilon$ 
would falsify the SETH. 

Much better additive approximations can be achieved for graphs with bounded (metric) parameters, including chordal graphs, HHD-free graphs, $k$-chordal graphs, and more generally all $\delta$-hyperbolic graphs (see \cite{CDHP01,CDK03,ChDr94,ChDEHV08,Chepoi2018FastEA,Dragan99-dam,Dr-chEAT,DraganG20-tcs,Terrain,Certif,DNB97}). For example, a vertex furthest from an arbitrary vertex has eccentricity at least $diam(G)-2$ for chordal graphs \cite{ChDr94} and at least $diam(G)-\lfloor k/2\rfloor$ for $k$-chordal graphs \cite{CDK03}.  Hence, the diameter in those graphs can be approximated within a small additive error in linear time by a BFS. In fact, the last vertex visited by a LexBFS has eccentricity at least $diam(G)-1$ for chordal graphs \cite{DNB97} as well as for HHD-free graphs \cite{Dragan99-dam}. Thus, although the existence of a subquadratic algorithm for computing the exact diameter of a chordal graph would falsify the SETH, a vertex with eccentricity at least $diam(G)-1$ can be found in linear time by a LexBFS. 
 Later, those results were generalized to all $\delta$-hyperbolic graphs~\cite{ChDEHV08,Chepoi2018FastEA,Terrain,Certif}. Note that chordal graphs and distance-hereditary graphs are 1-hyperbolic, while $k$-chordal graphs are $\lfloor k/2\rfloor/2$-hyperbolic~\cite{WZ11}. Gromov~\cite{Gromov} defines $\delta$-hyperbolic graphs via a simple 4-point condition:
for any four vertices $u,v,w,x$, the two largest of the three distance sums $d(u,v) + d(w,x)$, $d(u,w) + d(v,x)$, and $d(u,x) + d(v,w)$ differ by at most $2\delta \geq 0$.
Such graphs have become of recent interest due to the empirically established presence of a small hyperbolicity in many real-world networks. 
 For every $\delta$-hyperbolic graph, a vertex furthest from an arbitrary vertex has eccentricity at least  $diam(G)-2\delta$~\cite{ChDEHV08}.  Furthermore, for any $m$-edge $\delta$-hyperbolic graph $G$, a vertex with eccentricity at most $rad(G)+2\delta$ can be found in $O(\delta m)$ time and a vertex with eccentricity at most $rad(G)+3\delta$ can be found in $O(m)$ time~\cite{ChDEHV08,Chepoi2018FastEA,Terrain,Certif}. 
Three approximation schemes for all eccentricities were presented in~\cite{Terrain}:
an approximate eccentricity function $\hat{e}$, constructible in $O(\delta m)$ time,  which
satisfies $e(v) - 2\delta \leq \hat{e}(v) \leq e(v)$, for all $v \in V$, 
and two spanning trees $T$, one constructible in $O(\delta m)$ time and the other in $O(m)$ time,
which satisfy $e_G(v) \leq e_T(v) \leq e_G(v) + 4\delta + 1$ and
$e_G(v) \leq e_T(v) \leq e_G(v) + 6\delta$, for all $v \in V$, respectively.  
Our results of Section \ref{sec:one} and Section \ref{sec:two} (and some results of \cite{YuCh1991}) show that $\alpha_i$-metric graphs behave like $\delta$-hyperbolic graphs. In a forthcoming paper  \cite{alpha-to-hyperb}, indeed we show that the hyperbolicity $\delta$ of an $\alpha_i$-metric graph depends linearly on $i$. However, the constants in our Theorems \ref{th:DiamOfCenter}-\ref{th:all-ecc--appr} are better than those that can be obtained by combining the hyperbolicity result of \cite{alpha-to-hyperb} with the algorithmic results of \cite{ChDEHV08,Chepoi2018FastEA,Terrain} on radii, diameters and all eccentricities of $\delta$-hyperbolic graphs.  In \cite{alpha-to-hyperb}, we also introduce a natural generalization of an $\alpha_i$-metric, which we call a {\em ($\lambda,\mu$)-bow metric}:  namely, if two shortest paths $P(u,w)$ and $P(v,x)$ share a common shortest subpath $P(v,w)$ of length more than $\lambda$ (that is, they overlap by more than  $\lambda$),  then the distance between $u$ and $x$ is at least $d(u,v)+d(v,w)+d(w,x)-\mu$.  
$\delta$-Hyperbolic graphs are ($\delta, 2\delta$)-bow metric and $\alpha_i$-metric graphs are ($0,i$)-bow metric. 

$(\alpha_1,\Delta)$-Metric graphs form an important subclass of both $\alpha_1$-metric graphs and weakly bridged graphs, and they contain all chordal graphs and all plane triangulations with inner vertices of degree at least 7.
In~\cite{WG16}, it was shown that every $(\alpha_1,\Delta)$-metric graph  admits an eccentricity 2-approximating spanning tree, i.e., a spanning tree $T$ such that $e_T(v)-e_G(v)\le 2$ for every vertex $v$.  As a result, for a chordal graph, an additive 2-approximation of all eccentricities can be computed in linear time~\cite{Dr-chEAT}.  
Finding similar results for general $\alpha_1$-metric graphs was left as an open problem in~\cite{WG16}. 

\medskip
\noindent
{\bf Our Contribution.} 
We prove several new results on metric intervals, eccentricity function,  and centers in $\alpha_i$-metric graphs, and their algorithmic applications, thus answering open questions in the literature~\cite{Dr-chEAT,WG16,YuCh1991}.
To list our contributions, we need to introduce on our way some additional notations and terminology.

Section \ref{sec:one} is devoted to general $\alpha_i$-metric graphs ($i\ge 0$).  
The set $S_k(u, v) = \{x \in I(u, v) : d(u, x) = k\}$ is called a {\em slice} of the interval $I(u,v)$ where $0 \leq k \leq d(u,v)$. An interval $I(u,v)$ is said to be $\lambda$-thin if $d(x,y) \le  \lambda$ for all $x,y \in S_k(u,v)$, $0<k< d(u,v)$. The smallest integer $\lambda$ for which all intervals of $G$ are $\lambda$-thin is called the {\em interval thinness }  of $G$. 
We show first that, in $\alpha_i$-metric graphs $G$, the intervals are $(i+1)$-thin.

The disk of radius $r$ and center $v$ is defined as $\{ u \in V : d(u,v) \leq r \}$, and denoted by $D(v,r)$. Sometimes, $D(v,r)$ is called the $r$-neighborhood of $v$. In particular, $N[v] = D(v,1)$ and $N(v) = N[v] \setminus \{v\}$ denote the closed and open neighbourhoods of a vertex $v$, respectively. 
More generally, for any vertex-subset $S$ and a vertex $u$, we define $d(u,S) = \min_{v \in S} d(u,v), \ D(S,r) = \bigcup_{v \in S}D(v,r), \ N[S] = D(S,1)$ and $N(S) = N[S] \setminus S$. 
We say that a set of vertices $S\subseteq V$ of a graph $G=(V,E)$ is {\sl $d^k$-convex} if for every two vertices $x,y\in S$ with $d(x,y)\ge k\ge 0$, the entire interval $I(x,y)$ is in $S$. For $k\le 2$, this definition coincides with the usual definition of convex sets in graphs \cite{BolSol78,Ch86,SoCh1983}: $S$ is {\sl convex} if for every $x,y\in S$, the interval $I(x,y)$ is also in $S$. Clearly, the intersection of two  $d^k$-convex sets is also $d^k$-convex. 
We show that, in $\alpha_i$-metric graphs $G$, the disks (and, hence, the centers $C(G)$) are $d^{2i-1}$-convex. The main result of Section \ref{sec:centers} states that the diameter of the center $C(G)$ of $G$ is at most $3i+2$, thus answering a  question raised in \cite{YuCh1991}. 

Let $F_G(v)$ be the set of all vertices of $G$ that are most distant from $v$.
A pair $x,y$ is called a pair of mutually distant vertices if $e_G(x)=e_G(y)=d_G(x,y)$, i.e., $x\in F_G(y), y\in F_G(x)$. 
In Section \ref{sec:appr-rad-diam}, we show that a vertex $x$ that is most distant from an arbitrary vertex $z$ has eccentricity at least $diam(G)-3i-2$.  
Furthermore,  a middle vertex $c$ of any shortest path between $x$ and $y\in F(x)$ has eccentricity at most $rad(G)+ 4i+(i+1)/2+2$, and a middle vertex $c^*$ of any shortest path between any two mutually distant vertices has eccentricity at most $rad(G)+2i+1$. Additionally, all central vertices of $G$ are within a small distance from $c$ and $c^*$, namely, $C(G)\subseteq D(c^*, 4i +3)$ and $C(G)\subseteq D(c, 4i + (i+1)/2+2)$. Hence, an additive $O(i)$-approximation of the radius and of the diameter of an $\alpha_i$-metric graph $G$ with $m$ edges can be computed in $O(m)$ time.  
In Section \ref{sec:appr-all-ecc}, we present  
three approximation algorithms for all eccentricities:
an $O(i m)$ time eccentricity approximation $\hat{e}(v)$ based on the distances from any vertex to two mutually distant vertices, which
satisfies $e(v) - 3i-2 \leq \hat{e}(v) \leq e(v)$ for all $v \in V$, 
and two spanning trees $T$, one constructible in $O(i m)$ time and the other in $O(m)$ time,
which satisfy $e_G(v) \leq e_T(v) \leq e_G(v) + 4i + 3$ and
$e_G(v) \leq e_T(v) \leq e_G(v) + 7i+5$, respectively. Hence, an additive $O(i)$-approximation of all vertex eccentricities of an $\alpha_i$-metric graph $G$ with $m$ edges can be computed in $O(m)$ time.  

Section \ref{sec:two} is devoted to $\alpha_1$-metric graphs. 
The eccentricity  function $e(v)$ of a graph $G$ is said to be {\it  unimodal}, if for every non-central vertex $v$ of $G$ there is a neighbor $u\in N(v)$ such that $e(u)<e(v)$ (that is, every local minimum of the eccentricity function is a global minimum).
We show in Section \ref{sec:ecc} that the eccentricity function on $\alpha_1$-metric graphs is almost unimodal in the sense that the only non-central vertices that violate the unimodality (that is, do not have a neighbor with smaller eccentricity) must have their eccentricity equal to $rad(G)+1$ and their distance from $C(G)$ must be $2$. In other words, every local minimum of the eccentricity function on an  $\alpha_1$-metric graph $G$ is a global minimum or is at distance $2$ from a global minimum. 
Such behavior of the eccentricity function was observed earlier in chordal graphs \cite{WG16}, in distance-hereditary graphs \cite{DraganG20-tcs} and in all $(\alpha_1,\Delta)$-metric graphs \cite{WG16} (note also that in Helly graphs the eccentricity function is unimodal~\cite{DraganPhD}). This almost unimodality of the eccentricity function turns out to be very useful in locating a vertex with eccentricity at most $rad(G)+1$ in a gradient descent  fashion. In Section \ref{sec:centr}, using the convexity of the center $C(G)$ of an $\alpha_1$-metric graph $G$, we show that the diameter of $C(G)$ is at most $3$. This generalizes known results for chordal graphs \cite{Ch88} and for  $(\alpha_1,\Delta)$-metric graphs \cite{WG16}. In Section \ref{sec:finding-a-central-vertex}, we present a local-search algorithm for finding a central vertex of an arbitrary $\alpha_1$-metric graph in subquadratic time. 
Our algorithm even achieves linear runtime on $(\alpha_1,\Delta)$-metric graphs, thus answering an open question from~\cite{WG16}.  In Section \ref{sec:appr-all-ecc-i=1}, we show how to approximate efficiently all vertex eccentricities in $\alpha_1$-metric graphs. 

\section{General case of $\alpha_i$-metric graphs for arbitrary $i\ge 0$} \label{sec:one}

First we present an auxiliary lemma. 

\begin{lemma}  \label{lm:auxiliary-GD}
Let $G$ be an $\alpha_i$-metric graph, and let $u,v,x,y$ be vertices such that $x\in I(u,v)$, $d(u,x)=d(u,y)$, and $d(v,y)\le d(v,x)+k$. Then, $d(x,y)\le k+i+2$.
\end{lemma}
\begin{proof}
    Set $\nu = i + k + 2$, and suppose for the sake of contradiction $d(x,y) > \nu$. Without loss of generality, $k$ is the minimum value for which a counter-example can be found. We may also assume, without loss of generality, that $d(u,x)$ is minimized. Let $x' \in N(x) \cap I(x,u)$ and let $y' \in N(y) \cap I(y,u)$. Observe that $d(u,x') = d(u,y')$, $x' \in I(u,v)$ and $d(v,y') \leq d(v,y) + 1 \leq d(x,v) + 1 + k = d(x',v) + k$. Therefore, by minimality of $d(u,x)$, we have $d(x',y') \leq \nu$ (otherwise, we could replace $x,y$ with $x',y'$). Now, there are two cases to be considered:
    \begin{itemize}
        \item Case $d(y',x') < d(y',x)$. We also have $d(v,x) < d(v,x')$. Since $G$ is an $\alpha_i$-metric graph, it implies $d(v,y') \geq d(v,x) + 1 + d(x',y') - i$. Then, $d(x',y') \leq i - 1 + d(v,y') - d(v,x) = i + d(v,y') - d(v,x') \leq i + k$. However, $\nu \leq d(x,y) - 1 \leq d(x',y') + 1 \leq i + k +1 = \nu -1$.
        \item Case $d(y',x') \geq d(y',x)$. Then, $d(x,y) \leq 1 + d(x,y') \leq 1 + d(x',y') \leq \nu + 1$. It implies $d(x,y) = \nu+1$, $d(x,y') = d(x',y') = \nu$, and $y' \in I(x,y) \cap N(y)$. In particular, $d(x,y') < d(x,y)$. Furthermore, we claim that we have $d(v,y) < d(v,y')$. Indeed, suppose for the sake of contradiction $d(v,y') \leq d(v,y)$. In this situation, $d(v,y') \leq d(u,y') + d(v,y') - d(u,y') \leq d(u,y) - 1 + d(v,y) - d(u,y') \leq d(u,x) - 1 + d(v,x) + k - d(u,x') = d(u,v) - d(u,x') + k - 1 = d(v,x') + k - 1$. By minimality of $k$, we obtain $d(x',y') \leq i + (k - 1) + 2 = i + k + 1 < \nu$, and a contradiction arises. Therefore, as claimed, $d(v,y) < d(v,y')$. Since $G$ is an $\alpha_i$-metric graph, it implies $d(x,v) \geq d(x,y') + 1 + d(y,v) - i = d(x,y) + d(y,v) - i$. But then, $\nu = d(x,y) -1  \leq d(x,v) - d(y,v) + i -1 \leq i -1$. 
    \end{itemize} 
    In both cases, we derive a contradiction. \qed
\end{proof}

Lemma \ref{lm:auxiliary-GD} is helpful in proving that in $\alpha_i$-metric graphs the intervals are rather thin.  

\begin{lemma}  \label{lm:int-thin}
If $G$ is an $\alpha_i$-metric graph, then its interval thinness is at most $i+1$.  
\end{lemma}
\begin{proof}
    Let $u,v,x,y \in V$ be such that $x,y \in I(u,v)$, and $d(u,x) = d(u,y)$.
    Suppose for the sake of contradiction $d(x,y) > i+1$. By Lemma~\ref{lm:auxiliary-GD} (applied for $k=0$), we have $d(x,y) = i+2$. 
    We further assume, without loss of generality, that $d(u,x)$ is minimized.
    In particular, let $x' \in N(x) \cap I(u,x), \ y' \in N(y) \cap I(u,y)$. 
    By minimality of $d(u,x)$ we have $d(x',y') \leq i+1$.
    We claim that $d(x',y) \geq i+2$. 
    Indeed, if it were not the case, then $d(x',y) < d(x,y)$ and so, since we also have $d(v,x) < d(v,x')$, we would obtain $d(y,v) \geq d(y,x) + d(x,v) - i = d(x,v) +2 > d(x,v) = d(y,v)$, getting a contradiction.
    This proves as claimed that $d(x',y) \geq i+2$.
    It implies $d(x',y') \geq i+1$, and so $d(x',y') = i+1$.
    However, in this situation, $d(x',y') < d(x',y)$ and $d(v,y) < d(v,y')$.
    As a result, $d(v,x') \geq d(v,y') + d(x',y') - i = d(v,y') + 1 > d(v,y') = d(v,x')$, and a contradiction arises. \qed
\end{proof}

\subsection{Centers of $\alpha_i$-metric graphs} \label{sec:centers}
In this subsection we show that the diameter of the center of an $\alpha_i$-metric graph is at most $3i+2$, hereby providing  an answer to a question raised in \cite{YuCh1991} whether the diameter of the center of an $\alpha_i$-metric graph can be bounded by a linear function of $i$.  In~\cite{YuCh1991},  the following relation between the diameter and the radius of an $\alpha_i$-metric graph $G$ was proven: $2rad(G)\ge diam(G)\ge 2rad(G)-i-1$. 

First we show that all disks (and hence the center $C(G)$) of an $\alpha_i$-metric graph $G$ is $d^{2i-1}$-convex. 

\begin{lemma}  \label{lm:convexity}
Every disk of an $\alpha_i$-metric graph $G$ is $d^{2i-1}$-convex. In particular, the center $C(G)$ of an $\alpha_i$-metric graph $G$ is $d^{2i-1}$-convex. 
\end{lemma}
\begin{proof}
    Since $C(G) = \bigcap\{D(v,rad(G)) : v \in V\}$, it suffices to prove the $d^{2i-1}$-convexity of an arbitrary disk.
    Let $v,x,y \in V$ be such that $x,y \in D(v,r)$ for some $r \geq 0$ but $I(x,y) \not\subseteq D(v,r)$.
    Let $a_x \in I(x,y) \setminus D(v,r)$ be such that $d(x,a_x)$ is maximized, and let $b_x \in N(a_x) \cap I(a_x,y)$ be arbitrary.
    Note that $b_x \in D(x,r)$ by maximality of $d(x,a_x)$.
    In particular, $d(v,b_x) < d(v,a_x)$. We also have $d(x,a_x) < d(x,b_x)$ because $a_x,b_x \in I(x,y)$.
    Since $G$ is an $\alpha_i$-metric graph, $r \geq d(v,x) \geq d(v,a_x) + d(a_x,x) - i = r+1 + d(a_x,x) - i$.
    Therefore, $d(a_x,x) \leq i -1$.
    Now, let $a_y \in I(x,y) \setminus D(v,r)$ be such that $d(y,a_y)$ is maximized.
    We prove as before $d(a_y,y) \leq i-1$.
    Furthermore, $d(a_y,y) \geq d(a_x,y)$ by maximality of $d(a_y,y)$.
    As a result, $d(x,y) = d(x,a_x)+d(a_x,y) \leq d(x,a_x)+d(a_y,y) \leq 2i-2$. \qed
\end{proof}

Next auxiliary lemma is crucial in obtaining many results of this section.  

\begin{lemma}  \label{lm:main-ecc-ineq}
Let $G$ be an $\alpha_i$-metric graph.  For any $x,y,v \in V$ and any integer $k\in \{0,\dots,d(x,y)\}$, there is a vertex  $c \in S_k(x,y)$ such that $d(v,c) \leq \max\{d(v,x), d(v,y)\}-\min\{d(x,c), d(y,c)\}+ i$ and $d(v,c) \leq \max\{d(v,x),d(v,y)\} + i/2.$ 
For an arbitrary vertex $z\in I(x,y)$, we have $d(z,v)\leq \max\{d(x,v), d(y,v)\}-\min\{d(x,z), d(y,z)\}+2i+1$ and $d(z,v)\leq  \max\{d(x,v), d(y,v)\}+3i/2+1$. Furthermore, $e(z) \leq \max\{e(x), e(y)\}-\min\{d(x,z), d(y,z)\}+ 2i+1$ and 
 $e(z) \leq \max\{e(x),e(y)\} + 3i/2+1$ when $v\in F(z)$.
\end{lemma}

\begin{proof} Let $z\in S_k(x,y)$, for some $k\in\{0,\dots,d(x,y)\}$, and $c$ be a vertex of $S_k(x,y)$ closest to $v$. By Lemma \ref{lm:int-thin}, $d(c,z)\le i+1$. 

Consider a neighbor $c'$ of $c$ on a shortest path from $c$ to $v$. We have $d(x,c)<d(x,c')$ or $d(y,c)<d(y,c')$ since otherwise, when $d(x,c)\ge d(x,c')$ and $d(y,c)\ge d(y,c')$, $c'$ must belong to $S_k(x,y)$ and a contradiction with the choice of $c$ arises. Without loss of generality, assume $d(x,c)<d(x,c')$. Then $c\in I(x,c')$, and we can apply $\alpha_i$-metric property to $x,c,c',v$ and get  $d(x,v)\ge d(x,c)+d(c,v)-i$, i.e., $d(c,v)\le d(x,v)-d(x,c)+i\le \max\{d(x,v), d(y,v)\}-\min\{d(x,c), d(y,c)\}+i$.  By adding also the triangle inequality $d(c,v)\leq d(v,x)+d(x,c)$ to $d(c,v)\le d(x,v)-d(x,c)+i$, we get $d(c,v)\le d(x,v)+i/2\le \max\{d(x,v), d(y,v)\}+i/2$. 

For arbitrary $z\in S_k(x,y)$, as $d(z,c)\le i+1$, $d(x,c)=d(x,z)$, $d(y,c)=d(y,z)$, we get $d(z,v)\leq d(z,c)+d(c,v)\le i+1+d(x,v)-d(x,c)+i\le \max\{d(x,v), d(y,v)\}-\min\{d(x,z), d(y,z)\}+2i+1$ and $d(z,v)\leq d(z,c)+d(c,v)\le i+1+d(x,v)+i/2\le \max\{d(x,v), d(y,v)\}+i+i/2+1$. Applying both inequalities to the case in which $v$ is furthest from~$z$, we get  $e(z) = d(z,v) \leq \max\{e(x), e(y)\}-\min\{d(x,z), d(y,z)\}+ 2i+1$
and $e(z) \leq \max\{e(x),e(y)\} + i+i/2+1$.
\qed 
\end{proof} 

Lemma \ref{lm:main-ecc-ineq} has an immediate corollary. 

\begin{corollary}\label{cor:BoundOnEccentricityAny}
Let $G$ be an $\alpha_i$-metric graph. 
Any vertices $x,y \in V$ and $c \in I(x,y)$ satisfy $e(c) \leq \max\{e(x),  e(y)\} + 3i/2+1$.
However, if $d(x,y) \geq 4i+2$, then any vertex $c' \in I(x,y)$ with $d(x,c') \geq 2i+1$ and $d(y,c') \geq 2i+1$ satisfies $e(c') \leq \max\{e(x),  e(y)\}$. 
 Furthermore, if $d(x,y) > 4i+3$ then any vertex $c' \in I(x,y)$ with $d(x,c') > 2i+1$ and $d(y,c') > 2i+1$ satisfies $e(c') < \max\{e(x),  e(y)\}$. 
\end{corollary}

\begin{proof}
By Lemma~\ref{lm:main-ecc-ineq}, $e(c) \leq \max\{e(x), e(y)\} + 3i/2+1$.
Suppose that $d(x,y) \geq 4i+2$ and consider any vertex $c' \in I(x,y)$ satisfying $d(x,c') \geq 2i+1$ and $d(c',y) \geq 2i+1$.
By Lemma~\ref{lm:main-ecc-ineq}, $e(c') \leq  \max\{e(x), e(y)\}-\min\{d(x,c'), d(y,c')\}+ 2i+1$.
Hence, $e(c') \leq \max\{e(x),  e(y)\}$. 

Suppose now that $d(x,y) > 4i+ 3$, i.e., $d(x,y) \geq 4i + 4$.
Consider any vertex $c' \in I(x,y)$ satisfying $d(x,c') > 2i+1$ and $d(c',y) > 2i+1$.
By Lemma~\ref{lm:main-ecc-ineq}, $e(c') \leq \max\{e(x), e(y)\}-\min\{d(x,c'), d(y,c')\}+ 2i+1$.
Hence, $e(c') < \max\{e(x),  e(y)\}$. \qed
\end{proof}

Using Corollary \ref{cor:BoundOnEccentricityAny} one can easily prove that the diameter of the center $C(G)$ of an  $\alpha_i$-metric graph $G$ is at most $4i+3$. Below we show that the bound can be improved. 

\begin{theorem}\label{th:DiamOfCenter}
    If $G$ is an $\alpha_i$-metric graph, then $diam(C(G)) \leq 3i+2$.
\end{theorem}
\begin{proof}
Let us write $r = rad(G)$ in what follows.
Suppose by contradiction $diam(C(G)) > 3i+2$.
Since $C(G)$ is $d^{2i-1}$-convex, every diametral path of $C(G)$ must be fully in $C(G)$.
In particular, there exist $x,y \in C(G)$ such that $d(x,y) = 3i+3$ and $I(x,y) \subseteq C(G)$.
Furthermore, for every $u \in V$ such that $\max\{d(u,x),d(u,y)\} < r$, we obtain $I(x,y) \subseteq D(u,r-1)$ because the latter disk is also $d^{2i-1}$-convex.
Therefore, for every $z \in I(x,y)$, we must have $F(z) \subseteq F(x) \cup F(y)$.

Let $ab$ be an edge on a shortest $xy$-path such that $d(a,x) < d(b,x)$.
Assume $F(b) \not\subseteq F(a)$. Let $v \in F(b) \setminus F(a)$ be arbitrary.
Since $G$ is an $\alpha_i$-metric graph, $d(v,y) \geq d(v,b)+d(b,y)-i = r + \left(d(b,y) - i \right)$. Therefore, $d(b,y) \leq i$.
In the same way, if $F(a) \not\subseteq F(b)$, then $d(a,x) \leq i$.
By induction, we get $F(z) \subseteq F(x)$ ($F(z) \subseteq F(y)$, respectively) for every $z \in I(x,y)$ such that $d(y,z) \geq i+1$ ($d(x,z) \geq i+1$, respectively).
In particular, for every $t$ such that $i+1 \leq t \leq d(x,y) - i - 1$, we must have $F(z) \subseteq F(x) \cap F(y)$ for every $z \in S_t(x,y)$.

Note that the above properties are true not only for $x,y\in C(G)$ with $d(x,y) = 3i+3$ but also for every $x',y'\in C(G)$ with $d(x,y)\ge 2i-1$, as $d^{2i-1}$-convexity argument can still be used.  

Let $c \in I(x,y)$ be such that $F(c) \subseteq F(x) \cap F(y)$ and $k:=|F(c)|$ is minimized.
We claim that $k < |F(x) \cap F(y)|$.
Indeed, let $v \in F(x) \cap F(y)$ be arbitrary.
By Lemma~\ref{lm:main-ecc-ineq}, some vertex $c_v \in S_{i+1}(x,y)$ satisfies $d(c_v,v) \leq r - 1$.
Then, $F(c_v) \subseteq \left( F(x) \cap F(y) \right) \setminus \{v\}$, and $k \leq |F(c_v)| \leq |F(x) \cap F(y)|-1$ by minimality of $c$.
Then, let $y_c \in I(x,y)$ be such that $F(y_c) \cap F(x) \cap F(y) \subseteq F(c)$ and $d(x,y_c)$ is maximized (such a vertex must exist because $c \in I(x,y)$ satisfies that condition).
We have $y_c \neq y$ because $F(x) \cap F(y) \not\subseteq F(c)$.
Therefore, the maximality of $d(x,y_c)$ implies the existence of some $v \in \left(F(x) \cap F(y)\right) \setminus F(c)$ such that $d(v,y_c) = r - 1$.
Since $G$ is an $\alpha_i$-metric graph, $d(v,y) \geq d(v,y_c) + d(y_c,y) - i = r + \left(d(y_c,y) - i - 1 \right)$. As a result, $d(y_c,y) \leq i +1$.

Then, for every $z \in S_{i+1}(x,y_c)$, since we have $d(z,y_c) = d(x,y) - i-1 - d(y_c,y) \geq d(x,y) - 2i-2 = i+1$, we obtain $F(z) \subseteq F(x) \cap F(y) \cap F(y_c) \subseteq F(c)$.
By minimality of $k$, $F(z) = F(c)$.
However, let $v \in F(c)$ be arbitrary.
By Lemma~\ref{lm:main-ecc-ineq}, there exists some $c' \in S_{i+1}(x,y_c)$ such that $d(c',v) \leq r-1$, thus contradicting $F(c') = F(c)$. \qed
\end{proof}

 \subsection{Approximating radii and diameters of $\alpha_i$-metric graphs}\label{sec:appr-rad-diam}

 In this subsection, we show that a vertex with eccentricity at most $rad(G)+O(i)$ and a vertex with eccentricity at least $diam(G)-O(i)$ of an $\alpha_i$-metric graph $G$ can be found in (almost) linear time. 
 
 First we show that a middle vertex of any shortest path between two mutually distant vertices has eccentricity at most $rad(G)+2i+1$.   Furthermore, the distance between any two mutually distant vertices is at least $diam(G)-3i-2$. 
 
\begin{lemma} \label{lm:appr-rad}
 Let $G$ be an $\alpha_i$-metric graph, $x,y$ be a pair of mutually distant vertices of $G$ and $z$ be a middle vertex of an arbitrary shortest path connecting $x$ and $y$. Then, 
 $e(z)\le rad(G)+2i+1$. 
 Furthermore, there is a vertex $c$ in $S_{\lfloor d(x,y)/2\rfloor}(x,y)$ with $e(c)\le rad(G)+i$. 
\end{lemma}

\begin{proof} 
Let $r=rad(G)$ and $k=\lfloor d(x,y)/2\rfloor$. We need to show that for every vertex $z\in S_{k}(x,y)$, $e(z)\le r+2i+1$ holds. It will be sufficient to show that for a specially chosen vertex $c\in S_k(x,y)$, we have $e(c)\le r+i$. Then, since for any two vertices $u,v\in S_k(x,y)$, $d(u,v)\le i+1$ (see Lemma \ref{lm:int-thin}), we will get $e(z)\le d(z,c)+e(c)\leq r+2i+1.$    

Consider in $S_k(x,y)$ vertices $c$ with smallest eccentricity and among all those vertices pick that one whose  $|F(c)|$ is as small as possible. Let $v\in F(c)$ be a most distant vertex from $c$ and consider a neighbor $t$ of $c$ on a shortest path from $c$ to $v$.  As $x,y$ is a pair of mutually distant vertices, we have $d(x,y)\ge d(x,v)$ and  $d(x,y)\ge d(y,v)$.  We know also that $d(c,x)\le d(c,y)\le r$ as $d(x,y)\le 2r$. 

If $d(x,c)<d(x,t)$ then $c\in I(x,t).$ By the $\alpha_i$-metric property  applied to $x,c,t,v$, we get $d(x,y)\ge d(x,v)\ge d(x,c)+d(c,v)-i$. That is, $d(c,y)\ge d(c,v)-i$ and therefore $e(c)=d(c,v)\le d(c,y)+i\le r+i$. Similarly, if $d(y,c)<d(y,t)$ then  $e(c)=d(c,v)\le d(c,x)+i\le r+i$ must hold. So, in what follows, we may assume that $d(x,c)\ge d(x,t)$ and $d(y,c)\ge d(y,t)$, i.e.,  $t\in S_k(x,y)$ must hold. 

If $e(t)=e(c)+1$ then for every $s\in F(t)$, we have $c\in I(t,s)$ and $d(c,s)= d(c,v)$.
If $e(t)=e(c)$ then, by the choice of $c$, there must exist a vertex $s\in F(t)\setminus F(c)$ (as $v\in  F(c)\setminus F(t)$), and we have $c\in I(t,s)$ and $d(c,s)= d(c,v)-1$. 
In both cases, by the $\alpha_i$-metric property  applied to $c\in I(t,s)$ and $t\in I(c,v)$, we get
$d(s,v)\ge d(s,c)+d(v,c)-i\ge 2d(c,v)-1-i=2e(c)-1-i$. 

Now, since $d(s,v)\le diam(G)\le 2r$, if $e(c)\ge r+i+1$, we get $2r\ge d(s,v)\ge 2e(c)-1-i\ge 2r+2i+2-i-1=2r+i+1>2r$, which is impossible. 

This proves that $e(c)\le r+i$ and therefore $e(z)\le r+2i+1$ for every $z\in S_k(x,y)$. \qed
\end{proof} 

\begin{lemma} \label{lm:appr-diam}
 Let $G$ be an $\alpha_i$-metric graph and $x,y$ be a pair of mutually distant vertices of $G$. Then, 
$d(x,y)\ge 
2rad(G)-4i-3$ and $d(x,y)\ge diam(G)-3i-2$. Furthermore, any middle vertex $z$ of a shortest path between $x$ and $y$ satisfies $e(z)\le \lceil d(x,y)/2\rceil+2i+1$. 
\end{lemma}
\begin{proof} Let $z$ be a middle vertex of an arbitrary shortest path connecting $x$ and $y$, and let $v$ be a vertex furthest from $z$. By Lemma \ref{lm:main-ecc-ineq}, $d(z,v)\leq \max\{d(x,v), d(y,v)\}-\min\{d(x,z), d(y,z)\}+2i+1$. 
Since $x,y$ are two mutually distant vertices, $d(x,y)\ge \max\{d(x,v), d(y,v)\}$. 
Hence, $e(z)=d(z,v)\leq d(x,y)-\lfloor d(x,y)/2\rfloor+2i+1=\lceil d(x,y)/2\rceil+2i+1$. That is, $d(x,y)\ge 2e(z)-4i-3\ge 
2rad(G)-4i-3$.  

To prove $d(x,y)\ge diam(G)-3i-2$, consider a diametral pair of vertices $u,v$, i.e., with $d(u,v)=diam(G)$, and let $k=\lfloor d(x,y)/2\rfloor$. By Lemma \ref{lm:main-ecc-ineq}, there are vertices $v',u'\in S_k(x,y)$ such that $d(v',v)\leq \max\{d(x,v), d(y,v)\}-\min\{d(x,v'), d(y,v')\}+i$ and $d(u',u)\leq \max\{d(x,u), d(y,u)\}-\min\{d(x,u'), d(y,u')\}+i$.  By the triangle inequality and Lemma \ref{lm:int-thin}, $diam(G)=d(u,v)\le d(u,u')+d(u',v')+d(v',v)\le (\max\{d(x,v), d(y,v)\}-\lfloor d(x,y)/2\rfloor+i)+(i+1)+(\max\{d(x,u), d(y,u)\}-\lfloor d(x,y)/2\rfloor+i)=\max\{d(x,v), d(y,v)\}+\max\{d(x,u), d(y,u)\}-d(x,y)+3i+2.$ Since, $x,y$ are mutually distant vertices, $d(x,y)\ge \max\{d(x,v), d(y,v)\}$ and  $d(x,y)\ge \max\{d(x,u), d(y,u)\}$. Hence, $diam(G)\le d(x,y)+d(x,y)-d(x,y)+3i+2=d(x,y)+3i+2$. That is, $d(x,y)\ge diam(G)-3i-2$. 
 \qed
\end{proof}


For each vertex $v\in V\setminus C(G)$ of a graph $G$ we can define a
parameter $$loc(v)=\min\{d(v,x): x\in V, e(x)<e(v)\}$$ and call it
the {\em locality} of $v$. It shows how far from $v$ a vertex with a smaller eccentricity than that one of $v$ exists. In $\alpha_i$-metric graphs, the locality of each vertex is at most $i+1$.  

\begin{lemma} \label{lm:locality-for-i}
  Let $G$ be an $\alpha_i$-metric graph. Then, for every vertex $v$ in $V\setminus C(G)$, $loc(v)\leq i+1$. 
\end{lemma}

\begin{proof}   Let $x$ be a vertex with $e(x)=e(v)-1$ closest to $v$. Consider a neighbor $z$ of $x$ on an arbitrary shortest path from $x$ to $v$. Necessarily, $e(z)=e(x)+1=e(v)$. Consider a vertex $u\in F(z)$. We have $u\in F(x)$ and $x\in I(z,u)$, $z\in I(x,v)$. By the $\alpha_i$-metric property,
$d(v,u)\ge d(v,x)+d(x,u)-i=d(v,x)-i+e(x)$. As $e(v)\ge d(v,u)$, we get $e(v)\ge d(v,x)-i+e(x)= d(v,x)-i+e(v)-1$, i.e., $d(v,x)\le i+1$.
\qed
\end{proof}

In $\alpha_i$-metric graphs, the difference between the eccentricity of a vertex $v$ and the radius of $G$ shows how far vertex $v$ can be from the center $C(G)$ of $G$. 

\begin{lemma} \label{lm:dist-to-cent--for-i}
  Let $G$ be an $\alpha_i$-metric graph and $k$ be a positive integer. Then, for every vertex $v$ of $G$ with $e(v)\leq rad(G)+k$, $d(v,C(G))\leq k+i$. 
\end{lemma}

\begin{proof}   Let $x$ be a vertex from $C(G)$ closest to $v$. Consider a neighbor $z$ of $x$ on an arbitrary shortest path from $x$ to $v$. Necessarily, $e(z)=e(x)+1=rad(G)+1$. Consider a vertex $u\in F(z)$. We have $d(u,x)=rad(G)$ and $x\in I(z,u)$, $z\in I(x,v)$. By the $\alpha_i$-metric property,
$d(v,u)\ge d(v,x)+d(x,u)-i=d(v,x)-i+rad(G)$. As $e(v)\ge d(v,u)$ and $e(v)\le rad(G)+k$, we get $rad(G)+k\ge e(v)\ge d(v,x)-i+rad(G)$, i.e., $d(v,x)\le i+k$.
\qed
\end{proof}

As an immediate corollary of Lemma  \ref{lm:dist-to-cent--for-i} we get:

\begin{corollary} \label{cor:dist-to-cent--for-i}
  Let $G$ be an $\alpha_i$-metric graph. Then, for every vertex $v$ of $G$, $$d(v,C(G))+rad(G)\ge e(v)\geq d(v,C(G))+rad(G)-i.$$
\end{corollary}

\begin{proof}  The inequality $e(v)\le d(v,C(G))+rad(G)$ is true for any graph $G$ and any vertex $v$ by the triangle inequality.
If $G$ is an $\alpha_i$-metric graph  then, by Lemma \ref{lm:dist-to-cent--for-i}, $d(v,C(G))\leq e(v)-rad(G)+i$. \qed
\end{proof}

So, in $\alpha_i$-metric graphs,  to approximate the eccentricity of a vertex $v$ up-to an additive one-sided error $i$, one only needs to know $rad(G)$ and the distance $d(v,C(G))$. 

Now, instead of a pair of mutually distant vertices, we consider a vertex $v$ furthest from an arbitrary vertex. It turns out that its eccentricity is also close enough to $diam(G)$.   Furthermore, a middle vertex of any shortest path between that vertex $v$ and a vertex $u$ furthest from $v$ has eccentricity at most $rad(G)+O(i)$.  

\begin{lemma} \label{lm:ecc-of-furthest-vertex}
 Let $G$ be an $\alpha_i$-metric graph and $v$ be an arbitrary vertex. Then, for every $u\in F(v)$, $e(u)\ge 2rad(G)-2i-diam(C(G))\ge 2rad(G)-5i-2\ge diam(G)-5i-2$.  
\end{lemma}

\begin{proof}  By Corollary \ref{cor:dist-to-cent--for-i}, $e(v)\geq d(v,C(G))+rad(G)-i$ and $e(u)\geq d(u,C(G))+rad(G)-i$. Let $v',u'\in C(G)$ be vertices of $C(G)$ closest to $v$ and $u$, respectively. By the triangle inequality, $e(v)=d(v,u)\le d(v,v')+d(v',u')+d(u',u)$.  Combining two inequalities for $e(v)$, we get $d(v,C(G))+rad(G)-i\le e(v)\le d(v,C(G))+d(v',u')+d(u,C(G))$, i.e., $d(u,C(G))\ge  rad(G)-i-d(v',u')\ge  rad(G)-i-diam(C(G))$. Taking into account inequality  $e(u)\geq d(u,C(G))+rad(G)-i$ and Theorem \ref{th:DiamOfCenter}, we have $e(u)\geq 2rad(G)-2i-diam(C(G))\ge 2rad(G)-5i-2\ge diam(G)-5i-2$.
\qed
\end{proof} 

{\color{black}

The latter inequality in Lemma \ref{lm:ecc-of-furthest-vertex} can be improved if we do not involve in the proof the diameter of the center $C(G)$ but relay on the interval thinness.  We get the same bound as for mutually distant vertices (see Lemma \ref{lm:appr-diam}) on the eccentricity of a vertex most distant from an arbitrary vertex.

\begin{lemma} \label{lm:ecc-of-furthest-vertex++}
 Let $G$ be an $\alpha_i$-metric graph and $u$ be an arbitrary vertex. Then, for every $v\in F(u)$, $e(v)\ge diam(G)-3i-2$.  
\end{lemma}

\begin{proof} Let $x,y$ be a diametral pair of vertices of $G$, i.e., $d(x,y)=diam(G)$. We will show that $\max\{d(v,x), d(v,y)\}\ge diam(G)-3i-2$, hereby getting $e(v) \ge \max\{d(v,x), d(v,y)\}\ge diam(G)-3i-2$.  

Assume, by way of contradiction, that $\max\{d(v,x), d(v,y)\}\le diam(G)-3(i+1)$. 
Let $k= \lfloor (d(x,y)-3(i+1))/2\rfloor$ and $x',y'$ be  vertices from $S_k(v,u)$ closest to $x$ and $y$, respectively.  Note that, since $d(u,v)=e(u)\ge rad(G)$ and $k=\lfloor (diam(G)-3(i+1))/2\rfloor\le \lfloor (2rad(G)-3(i+1))/2\rfloor <rad(G)$, $k$ is smaller that $d(v,u)$ and hence vertices $x',y'$ exist. By Lemma  \ref{lm:int-thin}, $d(x',y')\le i+1$.  

Let $d(v,y)\le diam(G)-3(i+1)$ and consider a neighbor $y''$ of $y'$ on a shortest path from $y'$ to $y$. As $y''\notin S_k(v,u)$, $d(v,y'')>d(v,y')$ or $d(u,y'')>d(u,y')$. In the former case, by the $\alpha_i$-metric property applied to $v,y',y'',y$, we get $d(v,y)\ge d(v,y')+d(y',y)-i$, i.e., $d(y',y)\le d(v,y)-d(v,y')+i\le d(x,y)-3(i+1)-\lfloor (d(x,y)-3(i+1))/2\rfloor+i= \lceil (d(x,y)-3(i+1))/2\rceil+i.$ In the latter case (i.e., when $d(u,y'')>d(u,y')$), by the $\alpha_i$-metric property applied to $u,y',y'',y$, we get $d(u,y)\ge d(u,y')+d(y',y)-i$, i.e., $d(y',y)\le d(u,y)-d(u,y')+i$. We know that $d(u,v)\ge d(u,y)$ (as $v\in F(u)$) and $d(u,v)-d(u,y')=d(v,y')$. Hence, $d(y',y)\le d(v,y')+i=k+i\le \lceil (d(x,y)-3(i+1))/2\rceil+i.$ Thus, in either case, $d(y',y)\le \lceil (d(x,y)-3(i+1))/2\rceil+i.$ 

By symmetry, also $d(v,x)\le diam(G)-3(i+1)$ implies $d(x',x)\le \lceil (d(x,y)-3(i+1))/2\rceil+i.$ 
But then, by the triangle inequality, $d(x,y)\le d(x,x')+d(x',y')+d(y',y)\le 
\lceil (d(x,y)-3(i+1))/2\rceil+i +i+1+ \lceil (d(x,y)-3(i+1))/2\rceil+i \le d(x,y)-1$. The contradiction obtained shows that   $\max\{d(v,x), d(v,y)\}\ge diam(G)-3i-2$ must hold. \qed
\end{proof} 
}

\begin{lemma}\label{lm:eccMiddleOf@BFS}
Let $G$ be an $\alpha_i$-metric graph, $z \in V$, $x \in F(z)$, and $y \in F(x)$. 
Any vertex $c \in S_{\lfloor d(x,y)/2 \rfloor}(x,y)$ satisfies $e(c) \leq rad(G) + 4i+(i+1)/2+2$ and  {\color{black} $e(c) \leq \lceil d(x,y)/2\rceil+ 5i+3$.}  
\end{lemma}
\begin{proof}
Let $v \in F(c)$ be a furthest vertex from $c$. 
Since $y \in F(x)$, $d(x,v) \leq d(x,y)$. 
If also $d(y,v) \leq d(x,y)$  then, by Lemma~\ref{lm:main-ecc-ineq}, 
$e(c) = d(c,v) \leq \max\{d(v,x), d(v,y)\}-\min\{d(x,c), d(y,c)\}+ 2i+1\leq d(x,y) -\lfloor d(x,y)/2\rfloor+2i+1=\lceil d(x,y)/2\rceil+2i+1\le 
rad(G)+2i+1$.   

Let now $d(y,v) > d(x,y)\ge d(x,v)$. Again, by Lemma~\ref{lm:main-ecc-ineq},   
$e(c) = d(c,v) \leq d(y,v) - \lfloor d(x,y)/2 \rfloor  + 2i+1$. 
By Lemma \ref{lm:ecc-of-furthest-vertex}, $d(x,y)=e(x)\ge 2rad(G) - 5i-2\ge d(y,v)- 5i-2$. 
Therefore, $e(c) \leq d(y,v) -  \lfloor (2rad(G) - 5i-2)/2 \rfloor + 2i+1\leq 2rad(G) - rad(G) + 4i+(i+1)/2+2=rad(G) + 4i+(i+1)/2+2$. Also, 
{\color{black} by Lemma \ref{lm:ecc-of-furthest-vertex++}, $d(x,y)=e(x)\ge diam(G) - 3i-2\ge d(y,v)- 3i-2$. Hence, $e(c) \leq d(y,v) -  \lfloor d(x,y)/2 \rfloor + 2i+1\leq d(x,y)  + 3i+2 -\lfloor d(x,y)/2 \rfloor +2i+1= \lceil d(x,y)/2 \rceil + 5i+3$. 
}
\qed
\end{proof}

Next we show that all central vertices are close to a middle vertex~$c$ of a shortest path between vertices $x$ and $y$, provided that $x$ is furthest from some vertex and that $y$ is furthest from~$x$. Namely, $D(c, 4i + (i+1)/2+2) \supseteq C(G)$ holds.

\begin{lemma}\label{lm:centerToMiddle2BFS}
 Let $G$ be an $\alpha_i$-metric graph, and let $z\in V$, $x\in F(z)$ and $y\in F(x)$. Let also $c$ be a middle vertex of an arbitrary shortest path connecting $x$ and $y$. Then, 
$C(G)\subseteq D(c, 4i + (i+1)/2+2)$.
\end{lemma}
\begin{proof}
Consider an arbitrary vertex $u \in C(G)$.
By Lemma \ref{lm:main-ecc-ineq}, $d(c,u) \leq \max\{d(x,u), d(y,u)\} - \min\{d(x,c), d(y,c)\} + 2i+1$.
As $e(u)=rad(G)$, $\max\{d(x,u), d(y,u)\} \leq rad(G)$ holds. 
Since $c$ is a middle vertex of a shortest path between $x$ and $y$,  $\min\{d(x,c), d(y,c)\} = \lfloor d(x,y) / 2 \rfloor$. 
As $x \in F(z)$, by Lemma \ref{lm:ecc-of-furthest-vertex}, $d(x,y) = e(x) \geq 2rad(G) - 5i-2$. 
Hence, $d(c,u) \leq rad(G) - \lfloor (2rad(G) - 5i - 2) / 2 \rfloor + 2i+1 \leq 4i + (i+1)/2+ 2$. \qed
\end{proof}

A stronger result can be obtained for a middle vertex of a shortest path between two mutually distant vertices. 

\begin{lemma} \label{lm:centerToMiddleMDP}
 Let $G$ be an $\alpha_i$-metric graph and $x,y$ be a pair of mutually distant vertices of $G$. Let also $c$ be a middle vertex of an arbitrary shortest path connecting $x$ and $y$. Then, 
$C(G)\subseteq D(c, 4i +3)$.
\end{lemma}

\begin{proof}
The proof is analogous to that of Lemma~\ref{lm:centerToMiddle2BFS}.
However, since $x,y$ are mutually distant,  by Lemma \ref{lm:appr-diam}, $d(x,y) \geq 2rad(G) - 4i - 3$.
Hence, for any $u \in C(G)$,
$d(c,u) \leq rad(G) - \lfloor (2rad(G) - 4i - 3) / 2 \rfloor + 2i+1 \leq 4i + 3$. \qed
\end{proof}

There are several algorithmic implications of the results of this subsection.  For an arbitrary connected graph $G$ with $m$ edges and a given vertex $z\in V$, a vertex $x\in F(z)$ most distant from $z$ can be found in linear ($O(m)$) time by a {\em breadth-first-search} $BFS(z)$ started at $z$.
A pair of mutually distant vertices of an $\alpha_i$-metric graph can be computed in $O(im)$ total time as follows.  By Lemma~\ref{lm:ecc-of-furthest-vertex++}, if $x$ is a most distant vertex from an arbitrary vertex $z$ and $y$ is a most distant vertex from $x$, then $d(x,y)\geq diam(G)-3i-2$. Hence, using at most $O(i)$ {\em breadth-first-searches}, one can generate a sequence of vertices $x:=v_1,y:=v_2, v_3, \dots v_k$ with $k\leq 3i+4$  such that each $v_i$ is most distant from $v_{i-1}$ (with $v_0=z$) and $v_k$, $v_{k-1}$ are mutually distant vertices (the initial value $d(x,y)\geq diam(G)-3i-2$ can be improved at most $3i+2$ times).

We summarize algorithmic results of this section in the following theorem. 

\begin{theorem}\label{th:appr-rad-diam}
There is a linear $(O(m))$ time algorithm which finds vertices $v$ and $c$ of an $m$-edge $\alpha_i$-metric graph $G$ such that $e(v)\geq diam(G)-3i-2$, $e(c)\le rad(G) + 4i+(i+1)/2+2$ and $C(G)\subseteq D(c, 4i +(i+1)/2+2)$. Furthermore, there is an almost linear  $(O(im))$ time algorithm which finds 
a vertex $c$ of $G$ such that 
$e(c)\le rad(G) + 2i+1$ and $C(G)\subseteq D(c, 4i +3)$.
\end{theorem}


\begin{corollary} 
An additive $O(i)$-approximation of the radius and of the diameter of an $\alpha_i$-metric graph $G$ with $m$ edges can be computed in $O(m)$ time.  
\end{corollary}

\subsection{Approximating all eccentricities in $\alpha_i$-metric graphs}\label{sec:appr-all-ecc}

In this subsection, we show that the eccentricities of all vertices  of an $\alpha_i$-metric graph $G$ can be approximated with an additive one-sided error at most $O(i)$ in (almost) linear total time.

Interestingly, the distances from any vertex $v$ to two mutually distant vertices give a very good estimation on the eccentricity of $v$. 

\begin{lemma} \label{lm:eccFromMDP}
Let $G$ be an $\alpha_i$-metric graph and $x,y$ be a pair of mutually distant vertices of $G$.
Any vertex $v \in V$ satisfies $\max\{d(x,v), d(y,v)\} \leq e(v) \leq \max\{d(x,v), d(y,v)\} + 3i+2$.
\end{lemma}
\begin{proof}
The inequality $e(v) \geq \max\{d(x,v), d(y,v)\}$ holds for any three vertices by definition of eccentricity. 
To prove the upper bound on $e(v)$ for any $v \in V$, consider a furthest vertex $u \in F(v)$ and let $k=\lfloor d(x,y)/2\rfloor$.
Note that, as~$x$ and~$y$ are mutually distant, $d(x,y) \geq \max\{d(x,u), d(y,u)\}$.
%
By Lemma \ref{lm:main-ecc-ineq}, there are vertices $v',u'\in S_k(x,y)$ such that $d(v',v)\leq \max\{d(x,v), d(y,v)\}-\min\{d(x,v'), d(y,v')\}+i$ and $d(u',u)\leq \max\{d(x,u), d(y,u)\}-\min\{d(x,u'), d(y,u')\}+i$.  By the triangle inequality and Lemma \ref{lm:int-thin}, $e(v)=d(u,v)\le d(u,u')+d(u',v')+d(v',v)\le (\max\{d(x,v), d(y,v)\}-\lfloor d(x,y)/2\rfloor+i)+(i+1)+(\max\{d(x,u), d(y,u)\}-\lfloor d(x,y)/2\rfloor+i)=\max\{d(x,v), d(y,v)\}+\max\{d(x,u), d(y,u)\}-d(x,y)+3i+2\le \max\{d(x,v), d(y,v)\}+d(x,y)-d(x,y)+3i+2= \max\{d(x,v), d(y,v)\}+3i+2.$  \qed
\end{proof}

By Lemma~\ref{lm:eccFromMDP}, we get the following left-sided additive approximations of all vertex eccentricities. 
Let  $x,y$ be a  pair of mutually distant vertices of $G$. For every vertex~$v \in V$, set $\hat{e}(v) := \max\{d(x,v), d(y,v)\}$. 

\begin{theorem} \label{lm:all-ecc-left}
Let $G$ be an $\alpha_i$-metric graph with $m$-edges. 
There is an algorithm which in total almost linear $(O(im))$ time outputs for every vertex $v\in V$ an estimate $\hat{e}(v)$ of its eccentricity $e(v)$ such that  $e(v)- 3i-2 \leq \hat{e}(v) \leq {e}(v).$
\end{theorem}
If~the minimum integer $i$ for a graph $G$ so that $G$ is an $\alpha_i$-metric graph is known in advance, then we can transform $\hat{e}$ into a right-sided additive~$(3i+2)$-approximation by setting~$\hat{e}(v) := \max\{d(x,v), d(y,v)\} + 3i+2$. Unfortunately, for a given graph to find the minimum $i$ such that $G$ is an $\alpha_i$-metric graph is not an easy problem. 
We observe that even checking whether a graph is $\alpha_1$-metric is at least as hard as checking if a graph has an induced subgraph isomorphic to $C_4$. Indeed, take an arbitrary graph $G$ and add a universal vertex to it. Let the resulting graph be $G'$. Then, $G'$ is an $\alpha_1$-metric graph if and only if $G$ is $C_4$-free. 

\medskip

In what follows, we present two right-sided additive eccentricity approximation schemes for all vertices, using a notion of \emph{eccentricity approximating spanning tree} introduced in \cite{Prisner} and investigated in~\cite{Chepoi2018FastEA,Dr-chEAT,Terrain,WG16,Duc19-EAT}.  We get for $m$-edge $\alpha_i$-metric graphs 
a $O(m)$ time right-sided additive $(9i+5)$-approximation and a $O(im)$ time right-sided additive $(4i+2)$-approximation. 

A spanning tree $T$ of a graph $G$ is called an {\em eccentricity $k$-approximating spanning tree} if for every vertex $v$ of $G$  $e_T(v)\leq e_G(v)+ k$ holds~\cite{Prisner}. All  $(\alpha_1, \triangle)$-metric graphs (including chordal graphs and the underlying graphs of 7-systolic complexes)  admit eccentricity 2-approximating spanning trees~\cite{WG16}. 
An eccentricity 2-approximating spanning tree of a chordal graph can be computed in linear time~\cite{Dr-chEAT}. An eccentricity $k$-approximating spanning tree with minimum $k$ can be found in $O(nm)$ time for any $n$-vertex, $m$-edge graph $G$~\cite{Duc19-EAT}. It is also known~\cite{Chepoi2018FastEA,Terrain} that if~$G$ is a $\delta$-hyperbolic graph, then $G$ admits an eccentricity $(4\delta+1)$-approximating spanning tree constructible in $O(\delta m)$ time and an eccentricity $(6\delta)$-approximating spanning tree constructible in $O(m)$ time. 

\begin{lemma}\label{lm:eccApproxTreeByMDP}
Let $G$ be an $\alpha_i$-metric graph with $m$ edges.
If $c$ is a middle vertex of any shortest path between a pair $x,y$ of mutually distant vertices of $G$ and $T$ is a $BFS(c)$-tree of $G$, then, for every vertex $v$ of $G$, $e_G(v)\leq e_T(v)\leq e_G(v)+ 4i+3.$ That is, $G$ admits an eccentricity $(4i+3)$-approximating spanning tree constructible in $O(im)$ time. 
\end{lemma}

\begin{proof} Let $e_G(v)$ ($e_T(v)$) be the eccentricity of $v$ in $G$ (in $T$, respectively). 
The eccentricity in~$T$ of any vertex~$v$ can only increase compared to its eccentricity in~$G$. Hence, $e_G(v) \leq e_T(v)$.
By the triangle inequality and the fact that all graph distances from vertex $c$ are preserved in $T$, 
$e_T(v) \leq d_T(v,c) + e_T(c) = d_G(v,c) + e_G(c)$.
We know that $e_G(v) \geq \max\{d_G(y,v), d_G(x,v)\}$. By 
Lemma~\ref{lm:main-ecc-ineq}, also $d_G(v,c) - \max\{d_G(y,v),\ d_G(x,v) \} \leq 2i+1 - \min\{d_G(y,c),\ d_G(x,c)\}$ holds. Since $c$ is a middle vertex of a shortest path between $x$ and $y$, 
necessarily $\min\{d_G(y,c),\ d_G(x,c)\} = \lfloor d_G(x,y)/2 \rfloor$ and,
by Lemma~\ref{lm:appr-diam}, $e_G(c) \leq \lceil d_G(x,y)/2 \rceil + 2i+1$.
Combining all these, we get
$e_T(v) - e_G(v) \leq d_G(v,c) + e_G(c) - e_G(v) 
 			   \leq d_G(v,c) - \max\{d_G(y,v),          d_G(x,v)\} + e_G(c) 
  			 \leq 2i+1 - \min\{d_G(y,c),      
                        d_G(x,c)\} + e_G(c) 
  			 \leq 2i+1 - \lfloor d_G(x,y) / 2 
                     \rfloor + e_G(c) 
  			 \leq 2i+1 - \lfloor d_G(x,y) / 2 
                     \rfloor + \lceil d_G(x,y)/2 \rceil + 2i+1  
  			 \leq 4i + 3.$ \qed
\end{proof}
\begin{lemma}\label{lm:eccApproxTreeByFurthest}
Let $G$ be an $\alpha_i$-metric graph with $m$ edges, and let $z\in V$, $x\in F(z)$ and $y\in F(x)$.
If $c$ is a middle vertex of any shortest path between $x$ and $y$ and $T$ is a $BFS(c)$-tree of $G$, then, for every vertex $v$ of $G$, $e_G(v)\leq e_T(v)\leq e_G(v)+ 7i+5.$ That is, $G$ admits an eccentricity $(7i+5)$-approximating spanning tree constructible in $O(m)$ time. 
\end{lemma}
\begin{proof}
The proof follows the proof of Lemma~\ref{lm:eccApproxTreeByMDP} with one adjustment: 
replace the application of Lemma~\ref{lm:appr-diam} which yields $e_G(c) \leq \lceil d_G(x,y)/2 \rceil + 2i+1$ with Lemma \ref{lm:eccMiddleOf@BFS} which yields $e_G(c) \leq \lceil d_G(x,y)/2 \rceil + 5i+3$. Hence, $e_T(v) - e_G(v) \leq 2i+1 - \lfloor d_G(x,y) / 2 \rfloor + \lceil d_G(x,y)/2 \rceil + 5i+3 \leq 7i + 5$. \qed 
\end{proof}

Note that the eccentricities of all vertices in any tree $T=(V,U)$ can be computed in $O(|V|)$ total time. It is a folklore by now that for trees the following facts are true:
(1) The center $C(T)$ of any tree $T$ consists of one vertex or two adjacent vertices; (2) The center $C(T)$ and the radius $rad(T)$ of any tree $T$ can be found in linear time; (3) For every vertex $v\in V$, $e_T(v)=d_T(v,C(T))+rad(T)$.
Hence, using $BFS(C(T))$ on $T$ one can compute $d_T(v,C(T))$ for all $v\in V$ in total $O(|V|)$ time. Adding now $rad(T)$ to $d_T(v,C(T))$, one gets  $e_T(v)$ for all $v\in V$. Consequently, by Lemma~\ref{lm:eccApproxTreeByMDP}  and Lemma~\ref{lm:eccApproxTreeByFurthest}, we get the following additive approximations for the vertex eccentricities in $\alpha_i$-metric graphs.  

\begin{theorem} \label{th:all-ecc--appr}
Let $G$ be an $\alpha_i$-metric graph with $m$ edges. 
There is an algorithm which in total linear $(O(m))$ time outputs for every vertex $v\in V$ an estimate $\hat{e}(v)$ of its eccentricity $e(v)$ such that $e(v)\leq \hat{e}(v)\leq e(v)+ 7i+5.$ Furthermore, there is an algorithm which in total almost linear $(O(im))$ time outputs for every vertex $v\in V$ an estimate $\hat{e}(v)$ of its eccentricity $e(v)$ such that $e(v)\leq \hat{e}(v)\leq e(v)+ 4i+3.$
\end{theorem}

\begin{corollary} 
An additive $O(i)$-approximation of all vertex eccentricities of an $\alpha_i$-metric graph $G$ with $m$ edges can be computed in $O(m)$ time.  
\end{corollary}

\section{Graphs with $\alpha_1$-metric} \label{sec:two}
Now we concentrate on $\alpha_1$-metric graphs, which contain all chordal graphs and all plane triangulations with inner vertices of degree at least 7 (see, e.g.,  \cite{Ch86,Ch88,WG16,YuCh1991}). For them we get much sharper bounds.  It is known that for $\alpha_1$-metric graphs the following relation between the diameter and the radius holds: $2rad(G)\ge diam(G)\ge 2rad(G)-2$~\cite{YuCh1991}.

First we recall some known results and give an auxiliary lemma. 

\begin{lemma} [\cite{BaCh2003}]  \label{lm:1-hyp}
Let $G$ be an $\alpha_1$-metric graph. Let $x,y,v,u$ be vertices of $G$ such that
$v \in I(x,y)$, $x\in I(v,u)$, and $x$ and $v$ are adjacent. Then $d(u,y)=d(u,x) + d(v,y)$ holds if and only if there exist a neighbor $x'$ of $x$ in $I(x,u)$ and a neighbor $v'$ of $v$ in  $I(v,y)$ with $d(x',v')=2$; in particular, $x'$ and $v'$ lie on a common shortest path of $G$ between $u$ and $y$.
\end{lemma}

\begin{theorem} [\cite{YuCh1991}]\label{th:charact}
$G$ is an  $\alpha_1$-metric graph if and only if all disks $D(v,k)$ $(v\in V$, $k\geq 1)$ of $G$ are convex and $G$ does not contain the graph \WW\ from Fig. \ref{fig:forbid} as an isometric subgraph.
\end{theorem}

  \begin{figure}[htb]
    \begin{center} \vspace*{-37mm}
      \begin{minipage}[b]{16cm}
        \begin{center} 
          \hspace*{-28mm}
          \includegraphics[height=16cm]{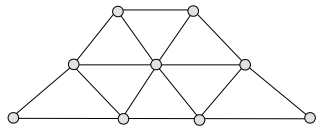}
        \end{center} \vspace*{-116mm}
        \caption{\label{fig:forbid} Forbidden isometric subgraph \WW.} %
      \end{minipage}
    \end{center}
   \vspace*{-5mm}
  \end{figure}

\begin{lemma} [\cite{SoCh1983}] \label{lm:obvious}
All disks $D(v,k)$ $(v\in V$, $k\geq 1)$ of a graph $G$ are convex if and only if for every vertices $x,y,z\in V$ and $v\in I(x,y)$, $d(v,z)\leq \max\{d(x,z),d(y,z)\}$.
\end{lemma}

Letting $z$ to be from $F(v)$, we get:

\begin{corollary} \label{cor:ecc-le-max}
If all disks $D(v,k)$ $(v\in V$, $k\geq 1)$ of a graph $G$ are convex then for every vertices $x,y\in V$ and $v\in I(x,y)$, $e(v)\leq \max\{e(x),e(y)\}$.
\end{corollary}

\begin{lemma} [\cite{WG16}] \label{lm:gen-eq}
  Let $G$ be an $\alpha_1$-metric graph and $x$ be an arbitrary
  vertex with $e(x)\geq rad(G)+1$.  Then, for every vertex $z\in
  F(x)$ and every neighbor $v$ of $x$ in $I(x,z)$, $e(v)\leq e(x)$
  holds.
\end{lemma}

We will need also the following auxiliary lemma.

\begin{lemma} \label{lm:equidistant--3-path}
Let $G$ be an $\alpha_1$-metric graph. Then, for every shortest path $P=(x_1,\dots,x_l)$ with $l\leq 4$ and a vertex $u$ of $G$ with $d(u,x_i)=k\ge 2$ for all $i\in \{1,\dots,l\}$,  there exists a vertex $u'$ at distance 2 from each $x_i$ $(i\in \{1,\dots,l\})$ and at distance $k-2$ from $u$.
\end{lemma}

\begin{proof} We prove by induction on $l$. If $l=1$, the statement is clearly correct. Assume now that there is a vertex $u'$ that is at distance 2 from each $x_i$ ($i\in \{1,\dots,l-1\}$) and at distance $k-2$ from $u$. Assume that $d(u',x_l)$ is greater than 2, i.e., $d(u',x_l)=3$. Consider a common neighbor $a$ of $u'$ and $x_{l-1}$. We have $x_{l-1}\in I(x_l,a)$ and $a\in I(u,x_{l-1})$. Then, by the $\alpha_1$-metric property, $k=d(x_l,u)\ge 1+k-1=k$, and therefore, by Lemma \ref{lm:1-hyp}, there must exist vertices $b$ and $c$ such that $bx_l, cb, ca\in E$ and $d(c,u)=k-2$. As $d(c,u)=d(u',u)=k-2$ and $d(u,a)=k-1$, by convexity of disk $D(u,k-2)$, vertices $u'$ and $c$ must be adjacent. If $c$ is at distance 2 from $x_1$  then, by convexity  of disk $D(c,2)$, each vertex $x_i$ ($1\le i\le l$) is at distance 2 from $c$, and we are done. If $c$ is at distance 3 from $x_1$, then from $u'\in I(c,x_1)$ and
$c\in I(u',x_l)$, by the $\alpha_1$-metric property, we get $d(x_1,x_l)\ge 2+2=4$, which is impossible since $d(x_1,x_l)\le 3$.
\qed
\end{proof}

\begin{corollary} \label{cor:C3orC5}
Let $G$ be an $\alpha_1$-metric graph. Then, for every edge $xy\in E$ and a vertex $u\in V$ with $d(u,x)=d(u,y)=k$,  either there is a common neighbor $u'$ of $x$ and $y$ at distance $k-1$ from $u$ or there exists a vertex $u'$ at distance 2 from $x$ and $y$ and at distance $k-2$ from $u$ such that, for every $z\in N(x)\cap N(u')$ and $w\in N(y)\cap N(u')$, the sequence $(x,z,u',w,y)$ forms an induced $C_5$ in $G$.
\end{corollary}

\begin{proof} We may assume that $k\ge 2$. By Lemma  \ref{lm:equidistant--3-path}, there exists a vertex $u'$ at distance 2 from $x$ and $y$ and at distance $k-2$ from $u$. Consider a common neighbor $z$ of $x$ and $u'$ and a common neighbor $w$ of $y$ and $u'$. If $zy,wx\notin E$ then,  by distance requirements, either $(x,z,w,y)$ induces a $C_4$ (which is impossible) or $(x,z,u',w,y)$ induces a $C_5$.
\qed
\end{proof}

\subsection{The eccentricity function on $\alpha_1$-metric graphs is almost unimodal} \label{sec:ecc}
The goal of this section is to prove the following theorem.

\begin{theorem}\label{th:ecc}
  Let $G$ be an $\alpha_1$-metric graph and $v$ be an
  arbitrary vertex of $G$.  If
  \begin{enumerate}
  \item[$(i)$] $e(v)> rad(G)+1$ or
  \item[$(ii)$] $e(v)=rad(G)+1$ and $diam(G)<2 rad(G)-1$,
  \end{enumerate}
  then there must exist a neighbor $w$ of $v$ with $e(w)<
  e(v)$. \\
  If $e(v)=rad(G)+k$ for some integer $k>0$,  then $d(v,C(G))\leq k+1$.
\end{theorem}

Theorem~\ref{th:ecc} says that if a vertex $v$ with
$loc(v)>1$ exists in an $\alpha_1$-metric graph $G$ then
$diam(G)\ge 2rad(G)-1$, $e(v)=rad(G)+1$ and $d(v,C(G))=2$.  That is,
only in the case when $diam(G)\in \{2rad(G)-1, 2rad(G)\}$, the
eccentricity function may fail to be unimodal 
and yet all local minima of the eccentricity function are concentrated around the center $C(G)$ of $G$ (they are at distance 2 from $C(G))$.
Two $\alpha_1$-metric graphs depicted in Fig. \ref{fig:sharp-ecc} show that this result is sharp.

  \begin{figure}[htb]
  \begin{center}%
    \vspace*{-40mm}
    \begin{minipage}[b]{16cm}
      \begin{center} 
        \hspace*{-48mm}
        \includegraphics[height=19cm]{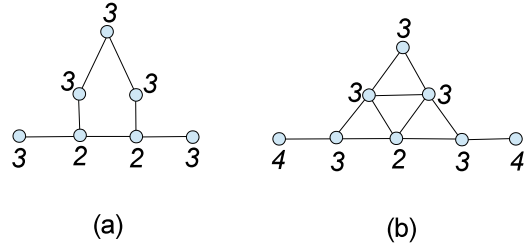}
      \end{center}%
      \vspace*{-122mm}
      \caption{\label{fig:sharp-ecc} Sharpness of the result of Theorem \ref{th:ecc}. (a) An $\alpha_1$-metric graph $G$ with $diam(G)=2rad(G)-1$ and a vertex (topmost) with locality 2. (b) A chordal graph (and hence an $\alpha_1$-metric graph) $G$ with $diam(G)=2rad(G)$ and a vertex (topmost) with locality 2. The number next to each vertex indicates its eccentricity.} %
    \end{minipage}
  \end{center}
  \vspace*{-6mm}
\end{figure}

We will split the proof of Theorem \ref{th:ecc} into a series of lemmas of independent interest. By Lemma \ref{lm:locality-for-i}, Lemma \ref{lm:dist-to-cent--for-i} and Corollary \ref{cor:dist-to-cent--for-i}, we already know that every vertex $v$ of an $\alpha_1$-metric graph has locality at most 2, is at distance at most $k+1$ from $C(G)$, provided that its eccentricity $e(v)$ is at most  $rad(G)+k$, and satisfies $d(v,C(G))+rad(G)\ge e(v)\geq d(v,C(G))+rad(G)-1.$
In the following lemmas, two specific properties of $\alpha_1$-metric graphs stated in Lemma \ref{lm:1-hyp} and Theorem \ref{th:charact} are heavily used. 

\begin{lemma} \label{lm:reduce-to-r-2}
  Let $G$ be an $\alpha_1$-metric graph and $v$ be a vertex of $G$ with $loc(v)=2$. Then, $e(v)\leq rad(G)+2$. Furthermore, if $e(v)=rad(G)+2$, then $diam(G)=2rad(G)$.
\end{lemma}

\begin{proof}   Let $k:=e(v)$ and $x$ be a vertex with $d(x,v)=2$ and $e(x)=k-1$ such that $|F(x)|$ is as small as possible.  Consider a common neighbor $z$ of $x$ and $v$ and a vertex $u\in F(z)$. Necessarily, $e(z)=e(v)=e(x)+1$ and $u\in F(x)$. We have $x\in I(z,u)$ and $z\in I(x,v)$. By the $\alpha_1$-metric property,
$d(v,u)\ge d(v,z)+d(x,u)=d(u,x)+1=e(z)=k$. As $e(v)\ge d(v,u)\ge e(z)=e(v)$, i.e., $k=d(v,u)$, by Lemma \ref{lm:1-hyp}, there must exist vertices $w$ and $t$ such that $wv, wt, tx\in E$ and $d(t,u)=d(u,x)-1=k-2$.

If $e(t)=e(x)+1$ then for every $s\in F(t)$, we have $x\in I(t,s)$ and $d(x,s)= k-1$.
If $e(t)=e(x)$ then, by the choice of $x$, there must exist a vertex $s\in F(t)\setminus F(x)$ (as $u\in  F(x)\setminus F(t)$), and we have $x\in I(t,s)$ and $d(x,s)= k-2$.
In both cases, by the $\alpha_1$-metric property  applied to $x\in I(t,s)$ and $t\in I(x,u)$, we get
$d(s,u)\ge d(s,x)+d(u,t)\ge k-2+k-2=2k-4$.

Since $d(s,u)\le diam(G)\le 2rad(G)$, we have $k\le rad(G)+2$ and if $k=rad(G)+2$ then $d(s,u)=diam(G)=2rad(G)$. \qed
\end{proof}

\begin{lemma} \label{lm:reduce-to-r-1}
  Let $G$ be an $\alpha_1$-metric graph and $v$ be a vertex of $G$ with $e(v)= rad(G)+2$. Then, $loc(v)=1$.
\end{lemma}

\begin{proof}  Assume, by way of contradiction, that $loc(v)=2$. Then, by Lemma \ref{lm:dist-to-cent--for-i}, $d(v,C(G))\le 3$. However, since $e(v)= rad(G)+2$,  $v$ cannot have a vertex from $C(G)$ at distance 2 or less as that would imply that a neighbor of $v$ on a shortest path to $C(G)$ has eccentricity at most $rad(G)+1$, contradicting with $loc(v)=2$. Thus, $d(v,C(G))= 3$.

Let $x$ be a vertex from $C(G)$ closest to $v$ and $P=(x,z,y,v)$ be a shortest path between $x$ and $v$ chosen in such a way that the neighbor  $y$ of $v$ in $P$  has $|F(y)|$ as small as possible. Necessarily, $e(x)=rad(G)=e(z)-1$ and $e(z)+1=e(y)=e(v)=rad(G)+2$. Consider a vertex $u\in F(y)$. Since $d(u,y)=rad(G)+2$, $d(u,x)\le rad(G)$ and $d(u,z)\le rad(G)+1$, we have $d(u,x)= rad(G)=d(u,z)-1=d(u,y)-2$, i.e., $x\in I(u,z)$. Applying the $\alpha_1$-metric property to $x\in I(u,z)$ and $z\in I(x,v)$, we get $d(v,u)\ge d(v,z)+d(x,u)= rad(G)+2$. By Lemma \ref{lm:1-hyp}, there exist vertices $f,w,t$ such that $fv,fz,fw,wt,tx\in E$ and $d(t,u)=rad(G)-1$. Notice that $f\neq y$ since $d(u,y)=rad(G)+2$ and $d(u,f)=rad(G)+1$. To avoid an induced $C_4$, vertices $f$ and $y$ must be adjacent. As $loc(v)=2$, we have also $e(f)\ge e(v)$.


Now, we have $f\in S_2(x,v)$ and $u\in F(y)\setminus F(f)$. By the choice of $y$, there must exist a vertex $s$ which is in $F(f)$ but not in $F(y)$.
Hence, $y\in I(f,s)$. Since also $f\in I(y,u)$, by the $\alpha_1$-metric property,  $d(u,s)\ge d(f,u)+d(s,y)\ge rad(G)+1+rad(G)+1=2rad(G)+2>diam(G)$, which is impossible.  Thus, $loc(v)=1$ must hold.  \qed
\end{proof}


\begin{lemma} \label{lm:new}
Let $G$ be an $\alpha_1$-metric graph. Let $v,c,b,f,a,u,s$ be vertices of $G$ such that $v,c,b,f,a$ form an induced $C_5$, $p=d(u,f)=d(v,u)-2=d(c,u)-2$ and $q=d(s,b)=d(v,s)-2=d(a,s)-2$ (See Fig. \ref{fig:rocket}). Then, either $d(u,s)=p+q+1$ or there is a vertex $h$ which is adjacent to all vertices of $C_5=(v,c,b,f,a)$.
\end{lemma}

\begin{proof}  By distance requirements, $f\in I(b,u)$ and $b\in I(f,s)$ hold. Hence, by $\alpha_1$-metric property, we  have $d(u,s)\ge p+q$.
Assume, in what follows, that  $d(u,s)= p+q$. Then, by Lemma  \ref{lm:1-hyp}, there must exist vertices $x,y,z$ such that $xf,xz,yb,yz\in E$ and $d(x,u)=p-1$ and $d(y,s)=q-1$. See Fig. \ref{fig:rocket} for an illustration.

  \begin{figure}[htb]
  \begin{center}%
    \vspace*{-30mm}
    \begin{minipage}[b]{16cm}
      \begin{center} 
        \hspace*{-36mm}
        \includegraphics[height=19cm]{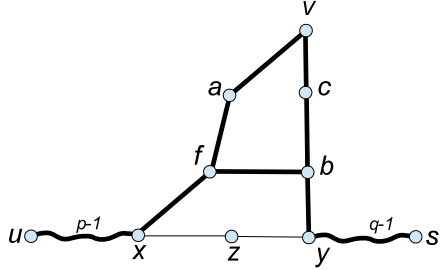}
      \end{center}%
      \vspace*{-130mm}
      \caption{\label{fig:rocket} Illustration to the proof of Lemma \ref{lm:new}.} %
    \end{minipage}
  \end{center}
  \vspace*{-6mm}
\end{figure}

As $d(v,x)=d(v,y)=3$ and $z\in I(x,y)$, by convexity of disk $D(v,3)$, we get $d(v,z)\leq 3$.
Vertex $z$ cannot be adjacent to $v,c,a$ as $d(v,y)=d(a,y)=3=d(v,x)=d(c,x)$. So, $2\le d(v,z)\leq 3$.

First consider the case when $d(v,z)=2$. Consider a common neighbor $h$ of $v$ and $z$. As $d(v,x)=d(v,y)=3$ and $d(v,f)=d(v,b)=d(v,z)=2$, by convexity of disk $D(v,2)$, we get $zf,zb\in E$. Convexities of disks $D(y,2)$ and $D(x,2)$ imply $hc, ha\in E$ (notice that $v$ is at distance 3 from both $y$ and $x$, vertices $a,h$ are at distance 2 from $x$, vertices $h,c$ are at distance 2 from $y$).
To avoid a forbidden induced $C_4$, $h$ must be adjacent to $f$ and $b$ as well. So, $h$ is adjacent to all vertices of $C_5=(v,c,b,f,a)$.

Now consider the case when $d(v,z)=3$. Consider a path $(v,h,g,z)$ between $v$ and $z$. If $z$ is adjacent to $f$ then it is adjacent to $b$ (to avoid an induced $C_4$), and wise versa. But if $fz, zb\in E$, by convexity of $D(z,2)$, $a$ and $c$ must be adjacent, contradicting with $C_5=(v,c,b,f,a)$
being an induced cycle. So, $zf,zb\notin E$, and hence $(x,z,y,b,f)$ forms an induced $C_5$ and $g\neq f,b$.

Vertices $c$ and $z$ cannot be at distance 2 from each other since then convexity of disk $D(c,2)$ and  $d(c,x)=3$  will imply $zf\in E$, which is impossible. Hence, $d(z,c)=3$ and, in particular, $gc\notin E$ and $h\neq c$. Similarly, $d(a,z)=3$ holds and, in particular, $ga\notin E$ and $h\neq a$.

We claim that $d(h,x)=d(h,y)=2$. If $d(h,y)=3$ then using also $d(c,z)=3$ we get $z\in I(h,y)$, $y\in I(z,c)$. By the $\alpha_1$-metric property, we obtain $d(h,c)\ge d(h,z)+d(c,y)= 4$, contradicting with $d(h,c)\le 2$. Similarly, $d(h,x)=2$ must hold.

Now, convexity of disks $D(y,2)$ and $D(x,2)$ gives first $hc,ha\in E$ (as $d(y,h)=d(y,c)=2=d(y,v)-1$ and $d(x,h)=d(x,a)=2=d(x,v)-1$) and then $hf,hb\in E$
(as $d(y,h)=d(y,f)=2=d(y,a)-1$ and $d(x,h)=d(x,b)=2=d(x,c)-1$). So, $h$ is adjacent to all vertices of $C_5=(v,c,b,f,a)$.  \qed
\end{proof}

\begin{lemma} \label{lm:reduce-to-2r-1}
  If an $\alpha_1$-metric graph $G$ has a vertex $v$ with $loc(v)>1$ and $e(v)=rad(G)+1$, then $diam(G)\ge 2rad(G)-1$.
\end{lemma}

\begin{proof}   Set $r:=rad(G)$. Assume, by way of contradiction, that $diam(G)\le 2r-2$.
Let $v$ be an arbitrary vertex with  $loc(v)>1$ and $e(v)=r+1$. Consider a vertex $s\in F(v)$ and a vertex $c\in S_1(v,s)$ with $|F(c)|$ as small as possible. Since $loc(v)>1$, $e(c)\ge e(v)$. Let also $u$ be an arbitrary vertex from $F(c)$.

We claim that $e(c)=d(u,c)=d(u,v)=e(v)$. If $d(u,c)>d(u,v)$, then $\alpha_1$-metric property applied to $v\in I(c,u)$ and $c\in I(s,v)$ gives $d(s,u)\ge d(s,c)+d(v,u)\ge e(v)-1+e(v)-1= 2r$, which is impossible. Hence, $e(c)=d(u,c)=d(u,v)=e(v)=r+1$ must hold.

Assume that a vertex $g$ exists in $G$ such that $gc,gv\in E$ and $d(g,u)=r=d(u,c)-1$. If $d(g,s)>d(c,s)$ then, by $\alpha_1$-metric property applied to $c\in I(g,s)$ and $g\in I(u,c)$, we get $d(s,u)\ge d(s,c)+d(g,u)= e(v)-1+r= 2r$, which is impossible. If $d(g,s)\le d(c,s)$, then $g\in S_1(v,s)$ and, by the choice of $c$, there must exist a vertex $t\in F(g)\setminus F(c)$ (recall that $u\in F(c)\setminus F(g)$ as $e(g)\ge r+1$). So, $\alpha_1$-metric property applied to $c\in I(g,t)$ and $g\in I(u,c)$, gives $d(t,u)\ge d(t,c)+d(g,u)\ge 2r$, which is impossible.

So, in what follows, we can assume that no common neighbor $g$ of $c$ and $v$ with $d(g,u)=r=d(u,c)-1$ can exist in $G$. Since, $d(u,c)=d(u,v)$, by Corollary \ref{cor:C3orC5}, there is a vertex $f$ which is at distance 2 from $c$ and $v$, at distance $d(u,c)-2=r-1$ from $u$ and forms with any $b\in N(c)\cap N(f)$ and any $a\in N(v)\cap N(f)$ an induced $C_5=(c,b,f,a,v)$.

We claim that $d(s,a)=d(s,v)$ for every $a\in N(v)\cap N(f)$. Indeed, if $d(s,a)<d(s,v)$ then, by convexity of disk $D(s,r)$, vertices $a$ and $c$ need to be adjacent (as both are adjacent to $v$ with $d(v,s)=r+1$), contradicting with $ac\notin E$. If $d(s,a)>d(s,v)$ then, $\alpha_1$-metric property applied to $v\in I(a,s)$ and $a\in I(u,v)$, gives $d(s,u)\ge d(s,v)+d(a,u)\ge 2r+1$, which is impossible.

From $d(s,a)=d(s,v)$, we have also $d(s,f)\ge d(s,c)$. Assume $d(s,f)=d(s,c)+1$. If $d(s,b)=d(s,c)$ then $b\in I(s,f)$. Since also $f\in I(b,u)$, $\alpha_1$-metric property implies $d(s,u)\ge d(s,b)+d(f,u)=d(s,c)+r-1=2r-1$, which is impossible. If now $d(s,b)=d(s,c)+1$, then $c\in I(s,b)$. Since also $b\in I(c,u)$, $\alpha_1$-metric property implies $d(s,u)\ge d(s,c)+d(b,u)=2r$, which is impossible. The last two contradictions show that $d(s,f)=d(s,c)+1$ is impossible, i.e., $d(s,f)= d(s,c)$ must hold. By convexity of disk $D(s,r)$, $d(s,b)\leq r$ holds for every $b\in I(f,c)$.  We distinguish between two cases. In both cases we get contradictions.

\medskip

\noindent
{\sl Case 1: There is a common neighbor $b$ of $f$ and $c$ which is at distance $r-1=d(s,c)-1$ from $s$.}

\medskip
We have $d(v,u)=r+1=d(c,u)=d(u,f)+2$ and $d(v,s)=r+1=d(a,s)=d(s,b)+2$. By Lemma \ref{lm:new}, either $d(u,s)=2r-1$ (which is impossible) or there is a vertex $h$ which is adjacent to all vertices of $C_5=(b,c,v,a,f)$.  The latter contradicts with our earlier claim that $d(s,a)=d(s,v)$ for every $a\in N(v)\cap N(f)$ (observe that $h\in N(v)\cap N(f)$ and $d(s,h)=d(s,b)+1=d(s,v)-1$).

\medskip

\noindent
{\sl Case 2: There is no common neighbor $b$ of $f$ and $c$ which is at distance $r-1=d(s,c)-1$ from $s$.}

\medskip
So, for every $b\in I(f,c)$, $d(b,s)=r$.  By Lemma \ref{lm:equidistant--3-path}, there exists a vertex $s'$ at distance 2 from $c$ and $f$ and at distance $r-2$ from $s$. Consider a common neighbor $x$ of $c$ and $s'$ and a common neighbor $y$ of $f$ and $s'$. We may assume that $yc, xf\notin E$  (otherwise, we are in Case 1). We claim that $d(y,v)=2$. Assume, $d(y,v)=3$. If also $d(a,x)=3$,  then $a\in I(y,v)$ and $v\in I(a,x)$ and $\alpha_1$-metric property implies $d(y,x)\ge d(y,a)+d(v,x)=4$, which is impossible. So, $d(a,x)=2$ must hold and therefore, by convexity of disk $D(a,2)$, vertices $x$ and $y$ are adjacent (observe that $x,y\in D(a,2)$ and $s'\notin D(a,2)$).
As $y\in I(f,x)$, $y\notin D(v,2)$ and  $f,x\in D(v,2)$, we get a contradiction with convexity of disk $D(v,2)$. Thus, $d(y,v)=2$ must hold. By convexity of disk $D(v,2)$, vertices $x$ and $y$ must be adjacent.

Consider a common neighbor $w$ of $y$ and $v$. By convexity of disk $D(s',2)$, vertices $w$ and $c$ are adjacent. To avoid an induced cycle $C_4$, $w$ and $x$ are also adjacent. If $d(w,u)=d(u,c)-1$ then, by the choice of $c$, there exists a vertex $t\in F(w)\setminus F(c)$ (recall that $u\in F(c)\setminus F(w)$ as $e(w)\ge r+1$). So, $\alpha_1$-metric property applied to $c\in I(w,t)$ and $w\in I(u,c)$, gives $d(t,u)\ge d(t,c)+d(w,u)\ge 2r$, which is impossible. So, $d(w,u)=d(u,c)=r+1$ must hold. In particular, vertices $w$ and $f$ cannot be adjacent. Note also that $d(a,y)=d(v,y)=2$ (as $d(a,s)=d(v,s)$)  and $wa\notin E$ (to avoid an induced $C_4$). Hence, $(v,w,y,f,a)$ induce a $C_5$ in $G$.

Now, we have $d(v,u)=r+1=d(w,u)=d(f,u)+2$ and $d(v,s)=r+1=d(a,s)=d(s,y)+2$. By Lemma \ref{lm:new}, either $d(u,s)=2r-1$ (which is impossible) or there is a vertex $h$ which is adjacent to all vertices of $C_5=(y,w,v,a,f)$.  The latter contradicts with our earlier claim that $d(s,a)=d(s,v)$ for every $a\in N(v)\cap N(f)$ (observe that $h\in N(v)\cap N(f)$ and $d(s,h)=d(s,y)+1=d(s,v)-1$).

The contradictions obtained prove the lemma.
 \qed
\end{proof}

Now Theorem  \ref{th:ecc} follows from Lemma \ref{lm:dist-to-cent--for-i},  Lemma \ref{lm:reduce-to-r-2},   Lemma \ref{lm:reduce-to-r-1} and
Lemma \ref{lm:reduce-to-2r-1}. Here we formulate three interesting corollaries of Theorem  \ref{th:ecc}.

\begin{corollary} \label{cor1:ofTh1}
  Let $G$ be an $\alpha_1$-metric graph. Then,
  \begin{enumerate}
  \item[$(i)$] if $diam(G)<2rad(G)-1$ $($i.e., $diam(G)=2rad(G)-2)$ then every local minimum of the eccentricity function on $G$ is a global minimum.
  \item[$(ii)$] if $diam(G)\ge 2 rad(G)-1$ then every local minimum of the eccentricity function on $G$ is a global minimum or is at distance $2$ from a global minimum.
  \end{enumerate}
\end{corollary}

\begin{corollary} \label{cor:ecc-func-formula}
For every $\alpha_1$-metric graph $G$ and any vertex $v$, the following formula is true:
$$d(v,C(G)) + rad(G) \ge e(v) \ge d(v,C(G)) + rad(G)- \epsilon,$$ where
$\epsilon\le 1$, if $diam(G) \ge 2rad(G)-1$, and $\epsilon=0$, otherwise.
\end{corollary}

A path $(v=v_0,\dots,v_k=x)$ of a graph $G$ from a vertex $v$ to a vertex $x$ is called {\em strictly decreasing} (with respect to the eccentricity function)  if for every $i$ ($0\le i\le k-1)$, $e(v_i)>e(v_{i+1})$. It is called {\em decreasing} if for every $i$ ($0\le i\le k-1)$, $e(v_i)\ge e(v_{i+1})$. An edge $ab\in E$ of a graph $G$ is called {\em horizontal} (with respect to the eccentricity function) if $e(a)=e(b)$.

\begin{corollary} \label{cor2:ofTh1}
  Let $G$ be an $\alpha_1$-metric graph and $v$ be its arbitrary vertex. Then, there is a shortest path $P(v,x)$ from $v$ to a closest vertex $x$ in $C(G)$  such that:
  \begin{enumerate}
  \item[$(i)$] if $diam(G)<2rad(G)-1$ $($i.e., $diam(G)=2rad(G)-2)$ then $P(v,x)$ is strictly decreasing;
  \item[$(ii)$] if $diam(G)\ge 2 rad(G)-1$ then $P(v,x)$ is decreasing and can have only one horizontal edge, with an end-vertex adjacent to $x$.
  \end{enumerate}
\end{corollary}

\subsection{Diameters of centers of $\alpha_1$-metric graphs} \label{sec:centr}
Observe that the center $C(G)$ of a graph
$G=(V,E)$ can be represented as the intersection of all the disks of
$G$ of radius $rad(G)$, i.e., $C(G)=\bigcap\{D(v,rad(G)): v\in V\}$.
Consequently, the center $C(G)$ of an
$\alpha_1$-metric graph $G$ is convex (in particular, it is
connected), as the intersection of
convex sets is always a convex set. In general, any set ${\cal
  C}_{\leq i}(G):=\{z\in V: ecc(z)\leq rad(G)+i\}$ is a convex set of
$G$ as ${\cal C}_{\leq i}(G)=\bigcap\{D(v,rad(G)+i): v\in V\}$.

\begin{corollary}\label{cl:ecclayer-convex}
  In an $\alpha_1$-metric graph $G$, all sets ${\cal C}_{\leq i}(G)$,
  $i \in \{0, \dots, diam(G)-rad(G)\}$, are convex.  In particular,
  the center $C(G)$ of an $\alpha_1$-metric graph $G$ is convex.
\end{corollary}

In this section, we provide sharp bounds on the diameter and the radius of
the center of an \mbox{$\alpha_1$-metric} graph.  Previously, it was
known that the diameter (the radius) of the center of a chordal graph
is at most 3 (at most 2, respectively)~\cite{Ch88}. To prove our result, we will need a few technical lemmas.

\begin{lemma} \label{lm:equidistant--4-path}
Let $G$ be an $\alpha_1$-metric graph. Then, for every shortest path $P=(x_1,x_2,x_3,x_4,x_5)$ and a vertex $u$ of $G$ with $d(u,x_i)=k$ for all $i\in \{1,\dots,5\}$, 
there exist vertices $t,w,s$ such that $d(t,u)=d(s,u)=k-1$, $k-2 \le d(w,u)\le k-1$, and $t$ is adjacent to $x_1,x_2,w$ and $s$ is adjacent to $x_4,x_5,w$.
\end{lemma}

\begin{proof} By Lemma \ref{lm:equidistant--3-path}, there is a vertex $u'$ that is at distance 2 from each $x_i$ $(1\le i\le 4)$ and at distance $k-2$ from $u$. Assume that $d(u',x_5)$ is greater than 2, i.e., $d(u',x_5)=3$. Consider a common neighbor $a$ of $u'$ and $x_{4}$. We have $x_{4}\in I(x_5,a)$ and $a\in I(u,x_{4})$. Then, by the $\alpha_1$-metric property, $k=d(x_5,u)\ge 1+k-1=k$, and therefore, by Lemma \ref{lm:1-hyp}, there must exist vertices $b$ and $c$ such that $bx_5, cb, ca\in E$ and $d(c,u)=k-2$. As $d(c,u)=d(u',u)=k-2$ and $d(u,a)=k-1$, by convexity of disk $D(u,k-2)$, vertices $u'$ and $c$ must be adjacent. If $c$ is at distance 2 from $x_1$  then, by convexity  of disk $D(c,2)$, each vertex $x_i$ ($1\le i\le 5$) is at distance 2 from $c$, and we can replace $u'$ with $c$ (see the case when $d(u',x_5)=2$ below). If $c$ is at distance 3 from $x_1$, then from $u'\in I(c,x_1)$ and
$c\in I(u',x_5)$, by the $\alpha_1$-metric property, we get $4=d(x_1,x_5)\ge 2+2=4$. Now, by Lemma \ref{lm:1-hyp}, there must exist vertices $t, w$ and $s$ such that $sx_5, sc, sw, tu', tx_1, tw\in E$. Furthermore, since $d(x_5,t)=d(x_5,x_2)=3$ and $d(x_5,x_1)=4$, by convexity of disk $D(x_5,3)$, vertices $t$ and $x_2$ must also be adjacent. Similarly, $sx_4\in E$. Necessarily, $d(u,t)=d(u,s)=k-1$. By convexity of disk $D(u,k-1)$, either $d(u,w)=k-1$ or $d(u,w)=k-2$. In the latter case, $w$ has all vertices $x_i$, $1\le i\le 5$, at distance 2, by convexity of disk $D(w,2)$.

The remaining case, when $d(u',x_5)=2$, can be handled in a similar (even simpler) manner. Set $w:=u'$, $t$ is any common neighbor of $u'$ and $x_1$ and $s$ is any common neighbor of $u'$ and $x_5$. By convexity of disks $D(x_1,3)$ and $D(x_5,3)$, the existence of edges $tx_2$ and $sx_4$ follows.
\qed
\end{proof}

  \begin{figure}[htb]
  \begin{center}%
    \vspace*{-25mm}
    \begin{minipage}[b]{16cm}
      \begin{center} 
        \hspace*{-36mm}
        \includegraphics[height=17cm]{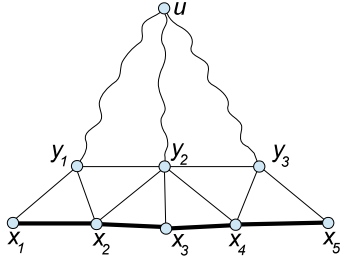}
      \end{center}%
      \vspace*{-120mm}
      \caption{\label{fig:PtoQ} Illustration to Lemma \ref{lm:PtoQ}.} %
    \end{minipage}
  \end{center}
  \vspace*{-6mm}
\end{figure}

\begin{lemma} \label{lm:PtoQ}
Let $G$ 
be an $\alpha_1$-metric graph. Then, for every shortest path $P=(x_1,x_2,x_3,x_4,x_5)$ and a vertex $u$ of $G$ with $d(u,x_i)=k$ for all $i\in \{1,\dots,5\}$,  there exists a shortest path $Q=(y_1,y_2,y_3)$ such that $d(u,y_i)=k-1$, for each $i\in \{1,\dots,3\}$, and $N(y_1)\cap P =\{x_1,x_2\}$,
$N(y_2)\cap P =\{x_2,x_3,x_4\}$ and $N(y_3)\cap P =\{x_4,x_5\}$ $($see Fig. \ref{fig:PtoQ}$)$.
\end{lemma}

\begin{proof} By Lemma \ref{lm:equidistant--4-path}, there exist vertices $t,s$ such that $d(t,u)=d(s,u)=k-1$, $t$ is adjacent to $x_1,x_2$ and $s$ is adjacent to $x_4,x_5$.  If $d(x_4,t)=3$, then $x_2\in I(t,x_4)$ and $t\in I(u,x_2)$ and, by the $\alpha_1$-metric property, $d(x_4,u)\ge 2+k-1=k+1$, which is impossible.  Hence, $d(x_4,t)=2$ must hold. Let $v$ be a common neighbor of $t$ and $x_4$. If $d(v,u)=k$, then again the $\alpha_1$-metric property applied to $v\in I(t,x_5)$ and $t\in I(u,v)$ gives $d(x_5,u)\ge 2+k-1=k+1$, which is impossible. Thus, $d(v,u)=k-1$ must hold. As $d(v,u)=k-1=d(s,u)$ and $d(u,x_4)=k$, by convexity of $D(u,k-1)$, vertices $v$ and $s$ must be adjacent. Notice also that $x_3$ cannot be adjacent to $t$ or to $s$. If, for example, $x_3s\in E$ (the case when $tx_3\in E$ can be handled in a similar way) then the $\alpha_1$-metric property applied to $s\in I(u,x_3)$ and $x_3\in I(s,x_1)$ will give $d(x_1,u)\ge 2+k-1=k+1$, which is impossible. Now, convexity of disk $D(x_1,2)$ implies edge $vx_3$ and convexity of disk $D(x_3,1)$ implies edge $vx_2$. Letting $y_1:=t$, $y_2:=v$ and $y_3:=s$, we get the required path $Q$. \qed
\end{proof}

\begin{theorem} \label{th:smaller-ecc}
Let $G$ be an  $\alpha_1$-metric graph. For every pair of vertices $s,t$ of $G$ with $d(s,t)\ge 4$ there exists a vertex $c\in I^o(s,t)$ such that $e(c)<\max\{e(s),e(t)\}$.
\end{theorem}

\begin{proof} It is sufficient to prove the statement for vertices $s,t$ with $d(s,t)=4$.

We know, by Corollary \ref{cor:ecc-le-max}, that $e(c)\le \max\{e(s),e(t)\}$ for every $c\in I(s,t)$.
  Assume, by way of contradiction, that there is no vertex $c\in I^o(s,t)$ such that $e(c)<\max\{e(s),e(t)\}$. Let, without loss of generality, $e(s)\le e(t)$. Then, for every $c\in I^o(s,t)$, $e(c)=e(t)$. Consider a vertex $c\in S_1(s,t)$. If $e(c)>e(s)$, then $e(c)=e(s)+1$. Consider a vertex $z$ from $F(c)$. Necessarily, $z\in F(s)$. Applying the $\alpha_1$-metric property to $c\in I(s,t)$, $s\in I(c,z)$, we get $e(c)=e(t)\ge d(t,z)\ge d(c,t)+d(s,z)= 3+e(s)=2+e(c)$, which is impossible.
  So, $e(s)=e(c)=e(t)$ for every $c\in I^o(s,t)$.

  Consider an arbitrary shortest path $P=(s=x_1,x_2,x_3,x_4,x_5=t)$ connecting vertices $s$ and $t$.
  We claim that for any vertex $u\in F(x_3)$ all vertices of $P$
  are at distance \mbox{$k:=d(u,x_3)=e(x_3)$} from $u$.  As $e(x_i)=e(x_3)$, we
  know that $d(u,x_i)\leq k$ ($1\leq i\leq 5$).
  Assume $d(u,x_i)=k-1$, $d(u,x_{i+1})=k$, and $i\leq 2$.
  Then, the $\alpha_1$-metric property applied to
  $x_i\in I(u,x_{i+1})$ and $x_{i+1}\in I(x_{i},x_{i+3})$ gives
  $d(x_{i+3},u)\geq k-1+2=k+1$, which is a
    contradiction with $d(u,x_{i+3})\leq k$.  So,
  $d(u,x_{1})= d(u,x_{2})= k$.  By symmetry, also $d(u,x_{4})=
  d(u,x_{5})= k$.
  Hence, by Lemma \ref{lm:PtoQ}, for the path $P=(x_1,x_2,x_3,x_4,x_5)$, there exists a shortest path $Q=(y_1,y_2,y_3)$ such that $d(u,y_i)=k-1$, for each $i\in \{1,\dots,3\}$, and $N(y_1)\cap P =\{x_1,x_2\}$,
$N(y_2)\cap P =\{x_2,x_3,x_4\}$ and $N(y_3)\cap P =\{x_4,x_5\}$  (see Fig. \ref{fig:PtoQ}). As $y_i\in I^o(x_1,x_5)=I^o(s,t)$ for each $i\in \{1,\dots,3\}$, we have $e(y_i)=e(x_3)=k$. 

 All the above holds for every shortest path
  $P=(s=x_1,x_2,x_3,x_4,x_5=t)$ connecting vertices $s$ and $t$.  Now,
  assume that $P$ is chosen in such a way that,
  among all vertices in $S_2(s,t)$, the vertex $x_3$ has the minimum number of furthest vertices, i.e., $|F(x_{3})|$ is as small as possible.
As $y_2$ also belongs to $S_2(s,t)$ and has $u$ at distance $k-1$, by the choice of $x_3$, there must exist a vertex
  $u'\in F(y_2)$ which is at distance $k-1$ from $x_3$.
  Applying the
  previous arguments to the path  $P':=(s=x_1,x_2,y_2,x_4,x_5=t)$, we will have $d(x_i,u')=d(y_2,u')=k$
  for $i=1,2,4,5$ and, by Lemma \ref{lm:PtoQ}, get two more vertices $v$ and
  $w$ at distance $k-1$ from $u'$ such that
  $vx_1,vx_2,wx_4,wx_5\in E$ and $vy_2,wy_2\notin E$ (see Fig. \ref{fig:smaller-ecc}).  By
  convexity of disk $D(u',k-1)$, also $vx_3,wx_3\in E$.
  Now consider the disk $D(x_2,2)$.
  Since $y_3,w$ are in the disk and $x_5$ is not, vertices $w$ and $y_3$
  must be adjacent.  But then vertices $y_2,x_3,w,y_3$ form a forbidden
  induced cycle $C_4$.

  \begin{figure}[htb]
  \begin{center}%
    \vspace*{-25mm}
    \begin{minipage}[b]{16cm}
      \begin{center} 
        \hspace*{-28mm}
        \includegraphics[height=17cm]{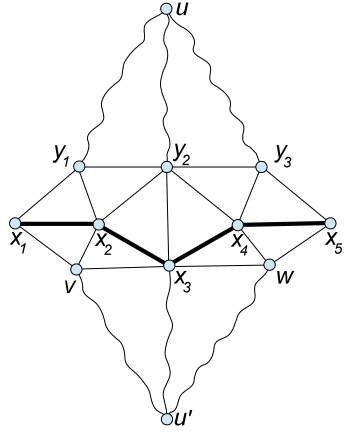}
      \end{center}%
      \vspace*{-95mm}
      \caption{\label{fig:smaller-ecc} Illustration to the proof of Theorem \ref{th:smaller-ecc}.} %
    \end{minipage}
  \end{center}
  \vspace*{-8mm}
\end{figure}

Thus, a vertex $c\in I^o(s,t)$ with $e(c)<\max\{e(s),e(t)\}$ must exist.
\qed
\end{proof}

\begin{corollary}\label{cor:diamcenter-our}
  Let $G$ be an  $\alpha_1$-metric graph.  Then, $diam(C(G))\leq 3$ and
  $rad(C(G))\leq 2$.
\end{corollary}

\begin{proof}   Assume, by way of contradiction, that there are vertices $s,t\in
  C(G)$ such that $d(s,t)\ge 4$. Then, by Theorem \ref{th:smaller-ecc}, there must exist a vertex $c\in I^o(s,t)$ such that $e(c)<\max\{e(s),e(t)\}$.
 The latter means that $e(c)<rad(G)$, which is impossible.
 So, $diam(C(G))\leq   3$ must hold. As $C(G)$ is a convex set of $G$, the subgraph of $G$ induced
  by $C(G)$ is also an $\alpha_1$-metric graph.  According
  to~\cite{YuCh1991}, $diam(G)\geq 2rad(G)-2$ holds for every
  $\alpha_1$-metric graph $G$.  Hence, for a graph induced by $C(G)$,
  we have $3\geq diam(C(G))\geq 2rad(C(G))-2$, i.e., $rad(C(G))\leq
  2$.
\qed
\end{proof}

As chordal graphs are $\alpha_1$-metric graphs, we get also the
following old result.

\begin{corollary}[\cite{Ch88}]\label{cor:diamcenter-ch}
  Let $G$ be a chordal graph.  Then, $diam(C(G))\leq 3$ and
  $rad(C(G))\leq 2$.
\end{corollary}

  \begin{figure}[htb]
  \begin{center}%
    \vspace*{-25mm}
    \begin{minipage}[b]{16cm}
      \begin{center} 
        \hspace*{-48mm}
        \includegraphics[height=19cm]{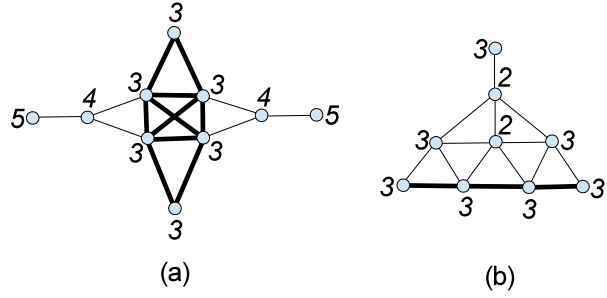}
      \end{center}%
      \vspace*{-115mm}
      \caption{\label{fig:sharp1} Sharpness of results of Theorem \ref{th:smaller-ecc} and Corollary \ref{cor:diamcenter-our}. (a) A chordal graph with $diam(C(G))=3$ and $rad(C(G))=2$. (b) An $\alpha_1$-metric graph with a pair of vertices at distance 3 for which no vertex with smaller eccentricity exists in a shortest path between them. The number next to each vertex indicates its eccentricity.} %
    \end{minipage}
  \end{center}
  \vspace*{-8mm}
\end{figure}

\subsection{Finding a central vertex of an $\alpha_1$-metric graph} \label{sec:finding-a-central-vertex}

We present a local-search algorithm in order to compute a central vertex of an arbitrary $\alpha_1$-metric graph in subquadratic time (Theorem~\ref{thm:compute-center-alpha1}).
Our algorithm even achieves linear runtime on an important subclass of $\alpha_1$-metric graphs, namely, $(\alpha_1,\Delta)$-metric graphs  (Theorem~\ref{thm:compute-center-alpha1-delta}), thus answering an open question from~\cite{WG16} where this subclass was introduced. 
The $(\alpha_1,\Delta)$-metric graphs are exactly the $\alpha_1$-metric graphs that further satisfy the so-called triangle condition: for every vertices $u,v,w$ such that $u$ and $v$ are adjacent, and $d(u,w) = d(v,w) = k$, there must exist some common neighbour $x \in N(u) \cap N(v)$ such that $d(x,w) = k-1$.
Chordal graphs, and plane triangulations with inner vertices of degree at least $7$, are $(\alpha_1,\Delta)$-metric graphs (see \cite{Ch86,Ch88,WG16,YuCh1991}).

\subsubsection{Distance-$k$ Gates.}

We first introduce the required new notations and terminology for this part.
In what follows, let ${\tt proj}(v,A) = \{a \in A : d(v,a) = d(v,A)\}$ denote the metric projection of a vertex $v$ to a vertex subset $A$.
For every $k$ such that $0 \leq k \leq d(v,A)$, we define $S_k(A,v) = \bigcup \{ S_k(a,v) : a \in {\tt proj}(v,A) \}$.
A {\em distance-$k$ gate} of $v$ with respect to $A$ is a vertex $v^*$ such that $v^* \in \bigcap\{ I(a,v) : a \in {\tt proj}(v,A) \}$ and $d(v^*,A) \leq k$.
If $k=1$, then following~\cite{ChDr94} we simply call it a gate.
Note that every vertex $v$ such that $d(v,A) \leq k$ is its own distance-$k$ gate.
We study the existence of distance-$k$ gates, for some $k \leq 2$, with respect to neighbour-sets and cliques. 
The latter are a cornerstone of our main algorithms. They are also, we think, of independent interest.

\begin{lemma}\label{lm:center-1}
Let $x$ and $v$ be vertices in an $\alpha_1$-metric graph $G$ such that $d(x,v) \geq 3$.
For every vertices $u,u' \in S_3(x,v)$, the metric projections ${\tt proj}(u,D(x,1))$ and ${\tt proj}(u',D(x,1))$ are comparable by inclusion.
\end{lemma}
\begin{proof}
Suppose by contradiction that $U = {\tt proj}(u,D(x,1))$ and $U' = {\tt proj}(u',D(x,1))$ are incomparable by inclusion.
Let $a \in U \setminus U'$ and $a' \in U' \setminus U$.
Observe that $d(u,a) < d(u',a)$ and similarly $d(u',a') < d(u,a')$.
Furthermore, $a,a' \in N(x) \cap I(x,v) (= D(x,1) \cap D(v,d(v,x)-1))$ must be adjacent because, according to Theorem~\ref{th:charact}, every disk is convex, and the intersection of two convex sets is also convex.
As a result, applying $\alpha_1$-metric property to $u,a,a',u'$, we get $d(u,u') \geq d(u,a) + d(a',u') = 4$. 
A contradiction arises because $u,u' \in S_3(x,v)$ and, by Lemma~\ref{lm:int-thin}, the interval thinness of $G$ is at most $2$. \qed
\end{proof}

\begin{lemma}\label{lm:center-2}
Let $x$ and $u$ be vertices in an $\alpha_1$-metric graph $G$ such that $d(x,u) = 3$.
For every vertices $w,w' \in N(u) \cap I(x,u)$, the metric projections ${\tt proj}(w,D(x,1)),{\tt proj}(w',D(x,1))$ are comparable by inclusion.
\end{lemma}
\begin{proof}
Since $w,w' \in N(u) \cap I(u,x)$, which is the intersection of two disks, and by Theorem~\ref{th:charact} every disk must be convex, the vertices $w$ and $w'$ must be adjacent.
Suppose by contradiction that ${\tt proj}(w,D(x,1)),{\tt proj}(w',D(x,1))$ are not comparable by inclusion.
Let $a,a' \in N(x)$ be such that $a \in N(w) \setminus N(w')$ and $a' \in N(w') \setminus N(w)$.
We prove similarly as above that $a,a' \in N(x) \cap I(x,u)$ must be adjacent.
But then, vertices $a,a',w',w$ induce a $C_4$, which is impossible. \qed 
\end{proof}

\begin{corollary}\label{cor:center-1}
Let $x$ be an arbitrary vertex of an $\alpha_1$-metric graph $G$. Every vertex $v$ of $G$ has a gate $v^*$ with respect to $D(x,1)$.
\end{corollary}
\begin{proof}
If $d(v,x) \leq 2$, then we can choose $v^* = v$.
From now on, $d(x,v) \geq 3$.
Let $u \in S_3(x,v)$ be such that $|{\tt proj}(u,D(x,1))|$ is maximized.
By Lemma~\ref{lm:center-1}, ${\tt proj}(v,D(x,1)) = {\tt proj}(u,D(x,1))$.
Then, let $v^* \in N(u) \cap I(u,x)$ be such that $|{\tt proj}(v^*,D(x,1))|$ is maximized.
By Lemma~\ref{lm:center-2}, ${\tt proj}(v^*,D(x,1)) = {\tt proj}(u,D(x,1)) = {\tt proj}(v,D(x,1))$. \qed
\end{proof}

We now turn our attention to cliques.
A difference appears between general $\alpha_1$-metric graphs and $(\alpha_1,\Delta)$-metric graphs, which partly justifies the better runtime achieved for computing a central vertex in the latter subclass.

\begin{lemma}[\cite{WG16}]\label{lm:dr-delta-clique}
Let $K$ be a clique in an $(\alpha_1,\Delta)$-metric graph $G$.
Every vertex $v$ has a gate $v^*$ with respect to $K$.
\end{lemma}

Lemma~\ref{lm:dr-delta-clique} does not hold for general $\alpha_1$-metric graphs.
For example, if one takes an edge in $C_5$, then the vertex at distance two to both end-vertices has no gate with respect to this edge. 
Nevertheless, we prove next the existence of {\em distance-two gates}.

\begin{lemma}\label{lm:center-4}
Let $K$ be a clique in an $\alpha_1$-metric graph $G$. 
Every vertex $v$ has a distance-two gate $v^*$ with respect to $K$.
\end{lemma}
\begin{proof}
We may assume, without loss of generality, $d(v,K) \geq 3$.
Let $v^* \in S_2(K,v)$ be maximizing $|{\tt proj}(v^*,K)|$.
Suppose by contradiction that ${\tt proj}(v^*,K) \neq {\tt proj}(v,K)$.
Let $x \in {\tt proj}(v,K) \setminus {\tt proj}(v^*,K)$ be arbitrary, and let $w \in S_2(x,v)$.
By maximality of $|{\tt proj}(v^*,K)|$, there exists a $y \in {\tt proj}(v^*,K) \setminus {\tt proj}(w,K)$.
Since $G$ is an $\alpha_1$-metric graph, $d(v^*,w) \geq d(v^*,y)+d(x,w) = 4$.
We also have $d(v^*,w) < d(v^*,y)+1+d(x,w) = 5$ because otherwise, the disk $D(v,d(v,K)-2)$ could not be convex, thus contradicting Theorem~\ref{th:charact}.
Therefore, $d(v^*,w)=4$.
By Lemma~\ref{lm:1-hyp}, there exist $y' \in N(y) \cap N(v^*)$ and $x' \in N(x) \cap N(w)$ such that $d(x',y') = 2$.
However, since we have $x',y' \in I(v^*,w) \setminus D(v,d(v,K)-2)$, the latter still contradicts the convexity of disk $D(v,d(v,K)-2)$. \qed
\end{proof}

\subsubsection{Computation of gates and distance-two gates.}

The problem of computing gates has already attracted some attention, {\it e.g.}, see~\cite{ChDr94}.
We use this routine in the design of our main algorithms.

\begin{lemma}[\cite{Duc23}]\label{lm:compute-gate}
Let $A$ be an arbitrary subset of vertices in some graph $G$ with $m$ edges.  
In total $O(m)$ time, we can map every vertex $v \notin A$ to some vertex $v^* \in D(v,d(v,A)-1) \cap N(A)$ such that $|N(v^*) \cap A|$ is maximized.
Furthermore, if $v$ has a gate with respect to $A$, then $v^*$ is a gate of $v$.    
\end{lemma}

The efficient computation of distance-two gates is more challenging.
We present a subquadratic-time procedure that only works in our special setting.

\begin{lemma}\label{lm:compute-distTwo-gate}
Let $K$ be a clique in some $\alpha_1$-metric graph $G$ with $m$ edges. 
In total $O(m^{1.41})$ time, we can map every vertex $v \notin K$ to some distance-two gate $v^*$ with respect to $K$.
Furthermore, in doing so we can also map $v^*$ to some independent set $J_K(v^*) \subseteq D(v^*,1)$ such that ${\tt proj}(v^*,K)$ is the disjoint union of neighbour-sets $N(w) \cap K$, for every $w \in J_K(v^*)$.    
\end{lemma}
\begin{proof}
For short, let us write $p_K(v) = |{\tt proj}(v,K)|$ for every vertex $v$.
If $v \in N(K)$, then we can set $v^* = v$, $J_K(v) = \{v\}$ and $p_K(v) = |N(v) \cap K|$.
This can be done in total $O(m)$ time for every vertex of $N(K)$.
Thus from now on we only consider vertices $v$ such that $d(v,K) \geq 2$. 

We compute $J_K(u)$ for every vertex $u$ such that $d(u,K) = 2$.
For that, we order the vertices of $N(u) \cap N(K)$ by nonincreasing $p_K$-value.
For every $x \in N(u) \cap N(K)$, if there is a triangle $xyu$ for some $y \in N(K)$ ordered before $x$, then we claim that $N(x) \cap K \subseteq N(y)$, and so we can ignore $x$.
Indeed, since $x,y$ are adjacent, $N(x) \cap K$ and $N(y) \cap K$ must be comparable by inclusion (otherwise, since $K$ is a clique, there would exist an induced $C_4$ with $x,y$ and arbitrary vertices of $(N(x) \setminus N(y)) \cap K, (N(y)\setminus N(x)) \cap K$).
Since $y$ was ordered before $x$, $p_K(y) \geq p_K(x)$, which implies as claimed $N(x) \cap K \subseteq N(y)$.
Else, for any $y \in N(K) \cap N(u)$ ordered before $x$, we claim that $N(x) \cap K$ and $N(y) \cap K$ must be disjoint.
Indeed, otherwise, there would be an induced $C_4$ with $u,x,y$ and an arbitrary vertex of $N(x) \cap N(y) \cap K$.
Then, we put vertex $x$ in $J_K(u)$.
Overall, we are left solving a variation of the well-known problem of deciding, for every edge in a graph, whether it is part of a triangle.
This problem can be solved in ${O}(m^{1.41})$ time~\cite{AYZ97}.

Finally, we consider all vertices $v$ such that $d(v,K) \geq 2$ by nondecreasing distance to $K$.
If $d(v,K) = 2$, then we set $v^* = v$ and $p_K(v) = \sum\{p_K(x) : x \in J_K(v)\}$.
Otherwise, we pick some neighbour $u \in N(v)$ such that $d(u,K) = d(v,K)-1$ and $p_K(u)$ is maximized.
Then, we set $v^* = u^*$ and $p_K(v) = p_K(u)$.
This procedure is correct assuming that a distance-two gate always exists, which follows from Lemma~\ref{lm:center-4}. \qed 
\end{proof}

\subsubsection{Local-search algorithms.}
Now that we proved the existence of gates and distance-two gates, and of efficient algorithms in order to compute them, we turn our attention to the following subproblem: being given a vertex $x$ in an $\alpha_1$-metric graph $G$, either compute a neighbour $y$ such that $e(y) < e(x)$, or assert that $x$ is a local minimum for the eccentricity function (but not necessarily a central vertex). 
Our analysis of the next algorithms essentially follows from the results of Section~\ref{sec:ecc}.
We first present the following special case, of independent interest, and for which we obtain a better runtime than for the more general Lemma~\ref{lm:loc-search-2}.

\begin{lemma}\label{lm:rad+2}
    Let $x$ be an arbitrary vertex in an $\alpha_1$-metric graph $G$ with $m$ edges.
    If $e(x) \geq rad(G)+2$, then $\bigcap\{N(x) \cap I(x,z) : z \in F(x)\} \neq \emptyset$, and every neighbour $y$ in this subset satisfies $e(y) < e(x)$.
    In particular, there is an $O(m)$-time algorithm that either outputs a $y \in N(x)$ such that $e(y) < e(x)$, or asserts that $e(x) \leq rad(G)+1$.
\end{lemma}
\begin{proof}
    If $e(x) \geq rad(G)+2$, then by Theorem~\ref{th:ecc}, there exists a $y \in N(x)$ such that $e(y) < e(x)$.
    In particular, $\bigcap\{N(x) \cap I(x,z) : z \in F(x)\} \neq \emptyset$.
    Let $y'$ be an arbitrary neighbour of $x$ in this subset. 
    Suppose by contradiction $e(y') \geq e(x)$.
    By Lemma~\ref{lm:gen-eq}, $e(y') \leq e(x)$.
    Hence, $e(x) = e(y') \geq rad(G) + 2$.
    Furthermore, $F(x) \cap F(y') = \emptyset$ because we assumed that $y' \in N(x) \cap I(x,z)$ for every $z \in F(x)$.
    But then, let $z \in F(x), \ z' \in F(y')$ be arbitrary.
    Since $G$ is an $\alpha_1$-metric graph, $d(z,z') \geq d(z,y')+d(x,z') = e(y') - 1 + e(x) - 1 \ge 2(rad(G)+1) > diam(G)$. The latter is impossible. 
    
    Finally, we describe our $O(m)$-time algorithm for an arbitrary vertex $x$ (of unknown eccentricity).
    We assume without loss of generality $e(x) \geq 2$.
    We compute gates $z^*$ with respect to $D(x,1)$ for every $z \in F(x)$, whose existence follows from Corollary~\ref{cor:center-1}.
    By Lemma~\ref{lm:compute-gate}, this can be done in $O(m)$ time.
    Then, we compute $K = \bigcap\{ N(x) \cap I(x,z) : z \in F(x) \}$.
    Note that $K = \bigcap\{ N(x) \cap N(z^*) : z \in F(x) \}$, and therefore, we can also compute $K$ in $O(m)$ time.
    If $K = \emptyset$, then we can assert $e(x) \leq rad(G)+1$.
    Else, let $y \in K$ be arbitrary.
    If $e(x) \leq e(y)$, then again we can assert that $e(x) \leq rad(G)+1$.
    Otherwise, we are done outputting $y$. \qed
\end{proof}

We can strengthen Lemma~\ref{lm:rad+2} as follows, at the expenses of a higher runtime.

\begin{lemma}\label{lm:loc-search-2}
    Let $x$ be an arbitrary vertex in an $\alpha_1$-metric graph $G$ with $m$ edges.
    There is an $O(m^{1.41})$-time algorithm that either outputs a $y \in N(x)$ such that $e(y) < e(x)$, or asserts that $x$ is a local minimum for the eccentricity function.
    If $G$ is $(\alpha_1,\Delta)$-metric, then its runtime can be lowered down to $O(m)$.
\end{lemma}
\begin{proof}
    If $e(x) \leq 2$ then $x$ is always a local minimum for the eccentricity function, unless $e(x) = 2$ and $x$ is adjacent to some universal vertex $y$.
    Therefore, we assume for the remainder of the proof $e(x) \geq 3$.
    First we compute $K = \bigcap\{ N(x) \cap I(x,z) : z \in F(x) \}$.
    This can be done in $O(m)$ time by using the exact same approach as for Lemma~\ref{lm:rad+2}.
    If $K = \emptyset$, then clearly $x$ is a local minimum for the eccentricity function.
    Otherwise, we are left deciding whether there exists a vertex of $K$ with eccentricity smaller than $e(x)$.
    For that, we claim that we only need to consider the vertices $v$ such that $d(v,K) \geq e(x)-1 \geq 2$.
    Indeed, by Theorem~\ref{th:charact}, the disks of $G$ must be convex, which implies that $K$ is a clique.
    In particular, every vertex $v$ such that $d(v,K) \leq e(x)-2$ is at a distance at most $e(x)-1$ to every vertex of $K$.
    Therefore, the claim is proved.    
    Now, if $d(v,K) = e(x)$ for at least one vertex $v$, then every vertex of $K$ must have eccentricity at least $e(x)$, hence we can assert that $x$ is a local minimum for the eccentricity function.
    Otherwise, let $F(K) = \{ v \in V : d(v,K) = e(x)-1 \}$.
    For every $y \in K$, in order to decide whether $e(y) < e(x)$, it suffices to decide whether $y \in \bigcap\{ {\tt proj}(v,K) : v \in F(K) \}$, or even more strongly to compute the number of vertices $v \in F(K)$ such that $y \in {\tt proj}(v,K)$.
    
    For general $\alpha_1$-metric graphs, we compute distance-two gates $v^*$ for every $v \in F(K)$, whose existence follows from Lemma~\ref{lm:center-4}.
    By Lemma~\ref{lm:compute-distTwo-gate}, this can be done in $O(m^{1.41})$ time.
    Being also given the independent sets $J_K(v^*) \subseteq N(v^*)$, for every $v \in F(K)$, as in Lemma~\ref{lm:compute-distTwo-gate}, we compute the following weight function $\alpha$ on $N(K)$: $\alpha(w) = |\{v \in F(K) : w \in J_K(v^*)\}|$.
    This can be done in $O(m)$ time.
    Recall that for every $v \in F(K)$, ${\tt proj}(v,K)$ is the disjoint union of $N(w) \cap K$ for every $w \in J_K(v^*)$.
    As a result, for every $y \in K$, the number of vertices $v \in F(K)$ such that $y \in {\tt proj}(v,K)$ is exactly $\sum\{ \alpha(w) : w \in N(y) \setminus K \}$, which can be calculated in total $O(m)$ time for every vertex of $K$.
    
    Finally, assume for the remainder of the proof that $G$ is $(\alpha_1,\Delta)$-metric, and let us modify this last part of the procedure as follows.
    We compute {\em gates} $v^*$ for every $v \in F(K)$, whose existence follows from Lemma~\ref{lm:dr-delta-clique}.
    By Lemma~\ref{lm:compute-gate}, this can be done in $O(m)$ time.
    For every $w \in N(K)$, let $\alpha(w) = |\{v \in F(K) : v^* = w\}|$.
    As before, for every $y \in K$, the number of vertices $v \in F(K)$ such that $y \in {\tt proj}(v,K)$ is exactly $\sum\{ \alpha(w) : w \in N(y) \setminus K \}$, which can still be calculated in total $O(m)$ time for every vertex of $K$. \qed
\end{proof}

\subsubsection{The main procedures.}
We start presenting our algorithm for the general $\alpha_1$-metric graphs. 

\begin{theorem}\label{thm:compute-center-alpha1}
    If $G$ is an $\alpha_1$-metric graph  with $m$ edges, then a vertex $x_0$ such that $e(x_0) \leq rad(G)+1$ can be computed in $O(m)$ time. 
    Furthermore, a central vertex can be computed in $O(m^{1.71})$ time.
\end{theorem}
\begin{proof}
    By Theorem~\ref{th:appr-rad-diam}, we can compute in $O(m)$ time a vertex $x_0$ such that $e(x_0) \leq rad(G)+3$.
    We repeatedly apply Lemma~\ref{lm:rad+2} until we can further assert that $e(x_0) \leq rad(G)+1$ (and, hence, by Theorem \ref{th:ecc}, $d(x_0,C(G))\le 2$). 
    Since there are at most two calls to this local-search procedure, the runtime is in $O(m)$.
    Then, we apply the following procedure, starting from $x_0$ and $X_0 := V$, until we can assert that the current vertex $x_i$ (considered at the $i^{th}$ iteration) is central.
    \begin{enumerate}
        \item\label{step-lowdeg} If $deg(x_i) \leq m^{.29}$, then we output a vertex of minimum eccentricity within $D(x_i,2)$, and then we stop. This is done by applying Lemma~\ref{lm:loc-search-2} to every vertex of $D(x_i,1)$. Otherwise ($deg(x_i) > m^{.29}$), we go to Step~\ref{step:smallerDisk}.
        \item\label{step:smallerDisk} Let $z_i \in F(x_i)$ be arbitrary. We set $X_{i+1} := X_i \cap D(x_i,5) \cap D(z_i,e(x_i)-1)$.
        If $X_{i+1} = \emptyset$, then we output $x_i$. Otherwise, we pick an arbitrary vertex $y \in X_{i+1}$, and then we go to Step~\ref{step:repair}.
        \item\label{step:repair} We consider several cases in what follows.
        \begin{enumerate}
            \item\label{step-repair-a} If $e(y) < e(x_i)$, then we output $y$;
            \item\label{step-repair-b} Else, if $e(y) = e(x_i)$, then we set $x_{i+1}:=y$ and we go back to Step~\ref{step-lowdeg};
            \item\label{step-repair-c} Else, there are three subcases. 
            If $\bigcap\{X_{i+1} \cap N(y) \cap I(y,w) : w \in F(y)\} = \emptyset$, then we output $x_i$.
            Otherwise, we pick an arbitrary neighbour $y'$ in this subset.
            If $e(y') \geq e(y)$, then we also output $x_i$.
            Else, we set $y:=y'$ and we repeat Step~\ref{step:repair}.
        \end{enumerate}
    \end{enumerate}
    We stress that at the $i^{th}$ iteration we ensure at Step~\ref{step:smallerDisk} that $x_i \in X_i \setminus X_{i+1}$.
    In particular, $X_0 \supset X_1 \supset \ldots \supset X_i$.
    It implies that our procedure eventually halts at some iteration $T$ because we always stop at Step~\ref{step:smallerDisk} if $X_{i+1} = \emptyset$.
    Suppose by contradiction that we do not output a central vertex.
    For every $i$ such that $0 \leq i < T$, we ensure at Step~\ref{step-repair-b} that $e(x_{i+1}) = e(x_i)$.
    Therefore, $e(x_T) = e(x_{T-1}) = \ldots = e(x_0) \leq rad(G)+1$.
    Since at the $T^{th}$ iteration, we output a vertex of eccentricity at most $e(x_T)$, we must have $e(x_T) \geq rad(G)+1$. 
    It implies $e(x_i) = rad(G)+1$ for every $i$ such that $0 \leq i \leq T$.
    Then, by Theorem~\ref{th:ecc}, $d(x_i,C(G)) \leq 2$ for every $i$ such that $0 \leq i \leq T$.
    Let us now consider the last step executed during iteration $T$.
    It cannot be Step~\ref{step-lowdeg} because a minimum eccentricity vertex within $D(x_T,2)$ must be central.
    Suppose by contradiction that we halt at Step~\ref{step:smallerDisk}.
    Then, $X_{T+1} = \emptyset$.
    We prove by induction that $C(G) \subseteq X_i$ for every $i$ with $0 \leq i \leq T+1$, obtaining a contradiction.
    This is true for $i=0$ because $X_0 = V$.
    Assume now this is true for some $i$ with $0 \leq i \leq T$.
    Recall that $d(x_{i},C(G)) \leq 2$.
    By Corollary~\ref{cor:diamcenter-our}, $diam(C(G)) \leq 3$.
    As a result, $C(G) \subseteq D(x_i,5)$.
    Furthermore, $C(G) \subseteq D(z_i,e(x_i)-1)$ because we have  $rad(G)=e(x_i)-1$.
    Hence, $C(G) \subseteq X_i \cap D(x_i,5) \cap D(z_i,e(x_i)-1) = X_{i+1}$.
    We deduce from the above that we must halt at Step~\ref{step:repair}.
    Since we suppose that we output a vertex of eccentricity $e(x_T) = rad(G)+1$, it implies that we halt at Step~\ref{step-repair-c}.
    In particular, we found a vertex $y \in X_{T+1}$ such that $e(y) \geq e(x_{T})+1$, and either $\bigcap\{X_{T+1} \cap N(y) \cap I(y,w) : w \in F(y)\} = \emptyset$, or $e(y') \geq e(y)$ for some arbitrary neighbour of $y$ in this subset.
    However, by Corollary~\ref{cor2:ofTh1}, there exists a shortest $yC(G)$-path that is decreasing and such that the only horizontal edge, if any, must have one end that is adjacent to the end-vertex in ${\tt proj}(y,C(G))$ (i.e., it must be the penultimate edge of the path, starting from $y$).
    In particular, the neighbour $y''$ of $y$ on this shortest path must satisfy $e(y'') < e(y)$.
    Since $y \in X_{T+1}$, $C(G) \subseteq X_{T+1}$ and according to Theorem~\ref{th:charact} the subset $X_{T+1}$ must be convex (as intersection of convex disks), we get that $y'' \in X_{T+1}$.
    Then, $\bigcap\{X_{T+1} \cap N(y) \cap I(y,w) : w \in F(y)\} \neq \emptyset$.
    Furthermore, since we assume $e(y) \geq rad(G)+2$, by Lemma~\ref{lm:rad+2}, every neighbour $y'$ in this subset must satisfy $e(y') < e(y)$.
    A contradiction occurs. 

    We end up analysing the runtime of the procedure.
    Recall that all vertices $x_0,x_1,\ldots,x_T$ are pairwise different.
    Since all such vertices, except maybe the last one, have degree more than $m^{.29}$, the number of iterations is in $O(m/m^{.29}) = O(m^{.71})$.
    Let us consider an arbitrary iteration $i$, for some $i$ such that $0 \leq i \leq T$.
    During Step~\ref{step-lowdeg}, either we do nothing or (only if $i=T$) we call Lemma~\ref{lm:loc-search-2} at most $O(m^{.29})$ times and then we stop.
    In the latter case, the runtime is in $O(m^{1.7})$.
    Step~\ref{step:smallerDisk} takes $O(m)$ time.
    Then, during Step~\ref{step:repair}, we only consider vertices $y \in D(x_i,5)$.
    Since every such vertex $y$ satisfies $e(y) \leq e(x_i)+5$, we can execute Step~\ref{step-repair-c} at most five times.
    For every execution of Step~\ref{step-repair-c}, we can slightly modify the algorithm of Lemma~\ref{lm:rad+2} in order to perform all computations in $O(m)$ time.
    Therefore, Step~\ref{step:repair} also takes $O(m)$ time.
    Overall, the total runtime of any iteration is in $O(m)$, except maybe for the last iteration whose runtime can be $O(m^{1.7})$.
    Since there are only $O(m^{.71})$ iterations, the final runtime is in $O(m^{1.71})$. \qed
\end{proof}

To lower the runtime to $O(m)$ for the $(\alpha_1,\Delta)$-metric graphs, we use a different approach that is based on the following additional properties of these graphs.
Unfortunately, these properties crucially depend on the triangle condition.

\begin{lemma}[\cite{WG16}]\label{lm:center-delta-1}
Let $G$ be an $(\alpha_1,\Delta)$-metric graph. 
Then, in every slice $S_k(y,z)$, there is a vertex $w$ that is universal to that slice, i.e., $S_k(y,z) \subseteq D(w,1)$.
\end{lemma}

\begin{lemma}\label{lm:center-delta-2}
Let $x,y$ be adjacent vertices in an $(\alpha_1,\Delta)$-metric graph $G$ such that both $x,y$ are local minima for the eccentricity function and $e(x)=e(y)=rad(G)+1$.
For every $z \in {\tt proj}(x,C(G))$, there exists a $u \in N(x) \cap N(z)$ such that $F(u) \subseteq F(x) \cap F(y)$.
\end{lemma}
\begin{proof}
First we prove that $z \in {\tt proj}(y,C(G))$.
Suppose by contradiction it is not the case.
By Theorem~\ref{th:ecc}, $d(x,C(G)) = d(y,C(G)) = 2$.
Therefore, $d(y,z) = 3$ and $x \in N(y) \cap I(y,z)$.
Let $w \in N(x) \cap N(z)$ be arbitrary.
Note that $e(w) \leq e(z)+1 = e(x)$.
It implies $e(w)=e(x)=rad(G)+1$ because $x$ is assumed to be a local minimum for the eccentricity function.
Then, let $v \in F(w)$ be arbitrary.
Since $G$ is $\alpha_1$-metric, $d(v,y) \geq d(v,z)+d(w,y) = rad(G)+2$. The latter is in contradiction with $e(y)=rad(G)+1$. 

Since $x,y$ are adjacent and $d(x,z) = d(y,z) = 2$, the triangle condition implies the existence of some common neighbour $u \in N(x) \cap N(y) \cap N(z)$.
Recall that $e(u) = e(x) = e(y) = rad(G)+1$.
Suppose by contradiction that there exists a $v \in F(u) \setminus F(x)$.
Since $G$ is an $\alpha_1$-metric graph, $d(z,v) \geq d(z,u)+d(x,v) = 1 + rad(G)$, which is impossible for $z\in C(G)$. 
As a result, $F(u) \subseteq F(x)$.
We prove similarly that $F(u) \subseteq F(y)$, and so $F(u) \subseteq F(x) \cap F(y)$. \qed 
\end{proof}

\begin{theorem}\label{thm:compute-center-alpha1-delta}
    If $G$ is an $(\alpha_1,\Delta)$-metric graph with $m$ edges, then a central vertex can be computed in $O(m)$ time.
\end{theorem}
\begin{proof}
    The algorithm starts from a vertex $x$ which is a local minimum for the eccentricity function of $G$.
    To compute such a vertex in $O(m)$ time, we first apply Theorem~\ref{th:appr-rad-diam} in order to compute a vertex of eccentricity at most $rad(G)+3$.
    Then, we apply Lemma~\ref{lm:loc-search-2} at most three times.

    We run a core procedure which either outputs two adjacent vertices $u,v \in D(x,1)$ such that $e(u)=e(v)=e(x)$ and $F(u),F(v)$ are not comparable by inclusion, or outputs a central vertex.
    In the former case, let $y \in F(u) \setminus F(v)$ and $z \in F(v) \setminus F(u)$ be arbitrary.
    Since $G$ is an $\alpha_1$-metric graph, $d(y,z) \geq d(y,v)+d(u,z) = 2e(x)-2$.
    If $d(y,z) \geq 2e(x)-1$, then $rad(G) > e(x)-1$, and therefore $x$ is a central vertex.
    From now on, we assume that $d(y,z) = 2e(x)-2$.
    Then, any vertex of eccentricity $e(x)-1$ must be contained in $S_{e(x)-1}(y,z)$.
    Let $w$ be such that $S_{e(x)-1}(y,z) \subseteq D(w,1)$, whose existence follows from Lemma~\ref{lm:center-delta-1}.
    We apply Lemma~\ref{lm:loc-search-2} in order to compute a minimum eccentricity vertex $x'$ within $D(w,1)$.
    Finally, we output any of $x,x'$ that has minimum eccentricity.
    To complete the description of our algorithm, we now present the core procedure.
    \begin{enumerate}
        \item We compute gates $z^*$ with respect to $D(x,1)$ for every $z \in F(x)$, whose existence follows from Corollary~\ref{cor:center-1}.
        By Lemma~\ref{lm:compute-gate}, this can be done in $O(m)$ time.
        Then, for every $y \in N(x)$, we compute $f_x(y) = |\{ z \in F(x) : y \in N(x) \cap I(x,z) \}|$.
        This can also be done in $O(m)$ time by enumerating the gates $z^*$ in $N(y)$ for every $y \in N(x)$.
        \item We choose a neighbour $y_1$ such that $f_x(y_1)$ is maximized. There are three cases.
        \begin{enumerate}
            \item If $e(y_1) > e(x)$ then, by Lemma~\ref{lm:gen-eq}, we can assert that $x$ is a central vertex.
            From now on, $e(y_1) \leq e(x)$.
            Since $x$ is a local minimum for the eccentricity function, $e(x) = e(y_1)$.
            \item If $y_1$ is not a local minimum for the eccentricity function of $G$, then by Lemma~\ref{lm:loc-search-2} we can compute a central vertex within $N(y_1)$ in $O(m)$ time. 
            From now on both $x,y_1$ are local minima for the eccentricity function.
            \item If $F(y_1) \setminus F(x) \neq \emptyset$, then we can output $u=x,v=y_1$.
            Indeed, we also have $F(x) \setminus F(y_1) \neq \emptyset$ because $y_1 \in N(x) \cap I(x,z)$ for some $z \in F(x)$.
            In what follows, we assume that $F(x),F(y_1)$ are comparable by inclusion, and so, $F(y_1) \subset F(x)$.
        \end{enumerate}
        \item Let $z_1 \in F(y_1)$ be arbitrary.
        We choose some neighbour $y_2$ within $N(x) \cap I(x,z_1)$ which maximizes $f_x(y_2)$.
        Note that in order to compute $N(x) \cap I(x,z_1)$, and so $y_2$, in $O(m)$ time, it suffices to only consider the vertices of $N(x) \cap N(z_1^*)$.
        There are two cases.
        \begin{enumerate}
            \item If $e(y_2) > e(x)$ then, by Lemma~\ref{lm:gen-eq}, we can assert that $x$ is a central vertex.
            From now on we assume $e(y_2) = e(x)$. 
            \item If $y_2$ is not a local minimum for the eccentricity function of $G$, then by Lemma~\ref{lm:loc-search-2} we can compute a central vertex within $N(y_2)$ in $O(m)$ time. 
            From now on vertices $x,y_1,y_2$ are local minima for the eccentricity function.
        \end{enumerate}
        \item Due to the maximality of $f_x(y_1)$ and the existence of a $z_1 \in F(x) \cap F(y_1)$, we have $f_x(y_2) \leq f_x(y_1) < |F(x)|$.
        Furthermore, both $F(x) \setminus F(y_1), F(x) \setminus F(y_2)$ are nonempty.
        For each $i \in \{1,2\}$, we compute $B_i = \bigcap\{ N(x) \cap I(x,z) : z \in F(x) \setminus F(y_i) \}$.
        Since $B_i = \bigcap\{ N(x) \cap N(z^*) :  z \in F(x) \setminus F(y_i) \}$, it can be done in $O(m)$ time.
        There are now two cases.
        \begin{enumerate}
            \item Assume the existence of an edge $uv$ where $u \in B_1, v \in B_2$. If  $\max\{e(u),e(v)\} > e(x)$ then, by Lemma~\ref{lm:gen-eq},  we can assert that $x$ is a central vertex. 
            Otherwise, we prove next that $F(u),F(v)$ are not comparable by inclusion, and therefore we can output $u,v$. 
            Indeed, $F(y_1),F(x) \cap F(y_2)$ are not comparable by inclusion because $$|F(x) \cap F(y_2)| = |F(x)| - f_x(y_2) \geq |F(x)| - f_x(y_1) = |F(y_1)|$$ 
            and $z_1 \in F(y_1) \setminus F(y_2)$.
            Furthermore, $F(x) \cap F(u) \subseteq F(y_1)$ ($F(x) \cap F(v) \subseteq F(y_2)$, respectively) because $u \in B_1$ ($v \in B_2$, respectively).
            Then, $F(u) \cap F(x) = F(y_1)$ and $F(v) \cap F(x) = F(y_2) \cap F(x)$, because otherwise this would contradict the maximality of either $f_x(y_1)$ or $f_x(y_2)$.
            \item Finally, assume that every vertex of $B_1$ is nonadjacent to every vertex of $B_2$.
            We claim that $x$ is a central vertex.
            Suppose by contradiction $e(x) = rad(G)+1$.
            Let $z \in {\tt proj}(x,C(G))$ be arbitrary.
            By Lemma~\ref{lm:center-delta-2}, there exist vertices $u,v \in N(x) \cap N(z)$ such that $F(u) \subseteq F(x) \cap F(y_1) = F(y_1)$, $F(v) \subseteq F(x) \cap F(y_2)$.
            Note that $\max\{e(u),e(v)\} \leq e(z) + 1 = rad(G)+1$.
            Since $x$ is a local minimum for the eccentricity function, we obtain $e(u)=e(v)=rad(G)+1$.
            Then, $u \in B_1$, $v \in B_2$.
            In particular, $u,v$ must be nonadjacent, thus contradicting the convexity of $D(x,1) \cap D(z,1)$, and so, Theorem~\ref{th:charact}.
        \end{enumerate}
    \end{enumerate} \qed 
\end{proof}

We leave as an open question whether there exists a linear-time algorithm for computing a central vertex in an $\alpha_1$-metric graph.

\subsection{Approximating all eccentricities in $\alpha_1$-metric graphs}\label{sec:appr-all-ecc-i=1}

It follows from the results of Section \ref{sec:finding-a-central-vertex} and Section \ref{sec:appr-rad-diam}
that a vertex  with eccentricity at most $rad(G)+1$ and a vertex  with eccentricity at least $diam(G)-5$ can be found in linear time for every $\alpha_1$-metric graph $G$. Furthermore, all vertex eccentricities with an additive one-sided error at most 5 in an $\alpha_1$-metric graph can be computed in total linear time (see Section \ref{sec:appr-all-ecc}).  Here, we present two immediate consequences of the results of Section  \ref{sec:two}, hereby answering some open questions from \cite{Dr-chEAT,WG16}. 

\begin{theorem}\label{thm:compute-ecc-alpha1}
    Every $\alpha_1$-metric graph  $G$ admits an eccentricity 3-approximating spanning tree. Furthermore, an additive 4-approximation of all vertex eccentricities in $G$ can be computed in subquadratic total time. 
\end{theorem}
\begin{proof}
 It is sufficient to show that any breadth-first-search tree $T$ of $G$ rooted at a  vertex $c\in C(G)$ is an eccentricity 4-approximating spanning tree of $G$ and any  breadth-first-search tree $T$ of $G$ rooted at a vertex $c\in C(C(G))$ (i.e., a central vertex of the subgraph of $G$ that is induced by $C(G)$)  is an eccentricity 3-approximating spanning tree of $G$.  We do not know how to find efficiently a vertex $c\in C(C(G))$ but, by   Theorem \ref{thm:compute-center-alpha1}, a central vertex $c$ of $G$ with $m$ edges can be computed in $O(m^{1.71})$ time. 

Consider an arbitrary vertex $v$ in $G$ and let $v'$ be a vertex of $C(G)$ closest to $v$. By  Corollary \ref{cor:ecc-func-formula},  $e(v) \ge d(v,v') + rad(G)- \epsilon,$ where
$\epsilon\le 1$. 
Since $T$ is a shortest path tree and $c$ is a central vertex of $G$,  $e_T(v)\le d_T(v,c)+e_T(c)=d_G(v,c)+e_G(c)=d_G(v,c)+rad(G)$. Hence, by the triangle inequality,
$e_T(v)-e_G(v)  \leq  d_G(v,c)+rad(G)-d_G(v,v') - rad(G) + \epsilon \leq   d_G(c,v')+ \epsilon\le d_G(c,v')+ 1.$

By Corollary \ref{cor:diamcenter-our}, we know $diam(C(G))\leq 3$ and
  $rad(C(G))\leq 2$. Hence, if $c\in C(G)$ then $d_G(c,v')\le 3$, and if $c\in C(C(G))$ then $d_G(c,v')\le 2$. \qed
\end{proof}

For $(\alpha_1,\Delta)$-metric graphs this result can be strengthened further.   
 
\begin{theorem}\label{thm:compute-ecc-alpha1-Delta}
     Every $(\alpha_1,\Delta)$-metric graph  $G$ admits an eccentricity 2-approximating spanning tree. Furthermore, an additive 3-approximation of all vertex eccentricities in $G$ with $m$ edges can be computed in total $O(m)$ time.
\end{theorem}
\begin{proof}
Let $T$ be a $BFS(c)$-tree, where $c$ is a central vertex of $G$. 
We can follow the proof of Theorem \ref{thm:compute-ecc-alpha1} and get 
$e_T(v)-e_G(v)  \leq   d_G(c,v')+ \epsilon$, where $v'$ is  a vertex of $C(G)$ closest to $v$.  In an $(\alpha_1,\Delta)$-metric graph  $G$, we have  
$\epsilon\le 1$, if $diam(G) = 2rad(G)$, and $\epsilon=0$, otherwise~\cite{WG16}.  Furthermore, when $diam(G)=2rad(G)$, $diam(C(G))\le 2$ and $rad(C(G))\le 1$ must hold~\cite{WG16}.  Thus, if $c\in C(G)$ then $e_T(v)-e_G(v)  \leq   d_G(c,v')+ \epsilon\le diam(C(G))+\epsilon\le 3$ (i.e., $\le 3+0$ or $\le 2+1$), and  if   $c\in C(C(G))$ then $e_T(v)-e_G(v)  \leq   d_G(c,v')+ \epsilon\le rad(C(G))+\epsilon\le 2$  (i.e., $\le 2+0$ or $\le 1+1$). 
Note also that a central vertex of an $(\alpha_1,\Delta)$-metric graph can be computed in linear time (Theorem  \ref{thm:compute-center-alpha1-delta}).  
\qed
\end{proof}

The existence of an eccentricity 2-approximating spanning tree in an $(\alpha_1,\Delta)$-metric graph is known already from \cite{WG16}. The second part of Theorem  \ref{thm:compute-ecc-alpha1-Delta} provides an answer to an open question from \cite{Dr-chEAT}.

\section{Concluding remarks}\label{sec:concl}

We conclude the paper with some immediate
questions building off our results.
\begin{itemize} 
\item[1.] Can our (approximation) bounds on eccentricities in general $\alpha_i$-metric graphs be improved? In particular, 
\begin{itemize} 
\item[i)] is our bound on the eccentricity of a middle vertex of a shortest path between two mutually distant vertices best possible? 
\item[ii)] is our bound on the eccentricity of a vertex furthest from an arbitrary vertex sharp?  
\item[iii)] can our bound $3i+2$ on the diameter of the center be improved to $2i+1$?
\end{itemize}
\item[2.] What best approximations of the radius and of the diameter of an $\alpha_1$-metric graph $G$ can be achieved in linear time? In particular, 
\begin{itemize} 
\item[iv)] does there exist a linear-time algorithm for finding a central vertex (and, hence, the exact radius) of an $\alpha_1$-metric graph? 
\item[v)] is it possible to show that the eccentricity of a vertex furthest from an arbitrary vertex is at least $diam(G)-2$?   
\item[vi)] can a vertex with eccentricity at least $diam(G)-1$ be found in linear time? 
\end{itemize}
 \end{itemize}
Recall that in chordal graphs (a subclass of $\alpha_1$-metric graphs) a central vertex can be found in linear time and the eccentricity of a vertex furthest from an arbitrary vertex is at least $diam(G)-2$ \cite{ChDr94}. Furthermore, a vertex with eccentricity at least $diam(G)-1$ can be found in linear time by a LexBFS \cite{DNB97}. On the other hand, computing 
the exact diameter of a chordal graph in subquadratic time is impossible unless the well known Strong Exponential Time Hypothesis (SETH) is false~\cite{BoCrHa16}.

\end{document}